\documentclass[11pt]{article}
\usepackage{fullpage}
\usepackage{graphicx}
\usepackage{epsfig}
\usepackage{graphics}
\usepackage{latexsym}
\usepackage{amsmath}
\usepackage{amsfonts}
\usepackage{amssymb}
\usepackage{mathrsfs}
\usepackage[switch]{lineno}
\usepackage{amsthm}
\usepackage{xspace}
\usepackage{epstopdf}
\usepackage{float}
\usepackage{hyperref}
\usepackage{pgfplots}
\usepackage{pythonhighlight}
\pgfplotsset{compat=1.8}

\numberwithin{equation}{section}
\usepackage[ruled,vlined]{algorithm2e}
\SetArgSty{textrm}

\usepackage[english]{babel}
\usepackage[nottoc]{tocbibind}

\usepackage{caption}
\usepackage{subcaption}

\usepackage{pifont}
\usepackage{booktabs}
\usepackage{multirow}

\newcommand{\redcomment}[1]{\textcolor{red}{\textrm{#1}}}

\usepackage{amsthm}     
\theoremstyle{plain}     
\newtheorem{theorem}{Theorem}

\newtheorem{corollary}{Corollary}

\newtheorem{lemma}{Lemma}
\newtheorem{proposition}{Proposition}
\theoremstyle{definition} 
\newtheorem{definition}{Definition}

\newtheorem{example}{Example}

\newtheorem{mechanism}{Mechanism}
\theoremstyle{remark} 
\newtheorem{remark}{Remark}

\makeatletter
\@addtoreset{equation}{section}
\def\section{\@startsection {section}{1}{\z@}{-3.5ex plus -1ex minus
 -.2ex}{2.3ex plus .2ex}{\large\bf}}
\makeatother

\def\bfm#1{\mbox{\boldmath$#1$}}

\def\x{\bfm x}

\def\0{\bfm 0}

\DeclareMathAlphabet{\mathpzc}{OT1}{pzc}{m}{it}

\newcounter{my}

\newcounter{my2}

\newcounter{my3}

\newcounter{my4}

\newcounter{my5}

\newcounter{my6}

\allowdisplaybreaks 

\begin{document}

\title{Mechanism Design for Connecting Regions Under Disruptions\thanks{A preliminary version  \cite{DBLP:conf/aaai/ChanLQ025} is in AAAI 2025.}}

\date{}
\maketitle

\vspace{-3em}
\begin{center}

\author{Hau Chan$^{1}$\quad Jianan Lin$^{2}$\quad Zining Qin$^{3}$\quad Chenhao Wang $^{4,3}$\\
${}$\\
$1$ University of Nebraska-Lincoln\\
$2$ Rensselaer Polytechnic Institute\\
$3$ Beijing Normal-Hong Kong Baptist University\\
$4$ Beijing Normal University-Zhuhai\\
\medskip
}

\end{center}

\begin{abstract}
Man-made and natural disruptions such as planned  constructions on roads, suspensions of bridges, and blocked roads by trees/mudslides/floods can often create obstacles that separate two connected regions. 
As a result, the traveling and reachability of agents from their respective regions to other regions can be affected. 
To minimize the impact of the obstacles and maintain agent accessibility, we initiate the problem of constructing a new pathway (e.g., a detour or new bridge) connecting the regions disconnected by obstacles from the mechanism design perspective.  
In the problem, each agent in their region has a private location and is required to access the other region. 
The cost of an agent 
is the distance from their location to the other region via the pathway.  
Our goal is to design strategyproof mechanisms that elicit truthful locations from the agents and approximately optimize the social or maximum cost of agents by determining  locations in the regions for building a pathway. 
We provide a characterization of all strategyproof and anonymous mechanisms. 
For the social and maximum costs, 
we provide upper and lower bounds on the approximation ratios of strategyproof mechanisms.
\end{abstract}

\section{Introduction}\label{sec:intro}

In modern societies, various types of infrastructures are constructed to connect  regions to facilitate the traveling or reachability of agents from their corresponding regions to other regions \cite{amekudzi2007transportation,forkenbrock1990economic,narayanaswami2017urban}. 
These types of infrastructures include highways, streets, roads, bridges, and transportation systems. 
For instance, using the road infrastructure, an agent from a region can drive to reach another region effectively. 

Unfortunately, these infrastructures can sometimes be interrupted either temporarily or permanently due to man-made or natural disruptions \cite{boakye2022role,faturechi2015measuring,gu2020performance,serdar2022urban}.
For example, man-made disruptions can refer to the planned large construction project of a road, the construction of a transportation hub (e.g., a subway station), the suspension of bridges (e.g., due to accidents), or the interruption of an area due to public activities (e.g., parades, temporary markets, or sporting events).  
In addition, natural disruptions can be in the form of the aftermath of disasters (e.g., earthquakes and storms), where roads and bridges are damaged or blocked by large trees,  mudslides, or floods. 

These disruptions can often result in obstacles that disconnect any two regions and affect the traveling and reachability of agents. 
Therefore, our goal is to determine the best way to construct new routes/pathways connecting the disconnected regions in order to minimize the impact of the obstacles and maintain the accessibility of the agents. 
In temporary disruptions with obstacles (e.g., road constructions, public events, or large trees on roads), the new pathways can be viewed as detours connecting the regions so that agents can continue to access other regions before the removal of the obstacles.
In permanent disruptions with obstacles (e.g., the suspensions of bridges or roads), the new pathways can be regarded as part of newly added roads or bridges connecting the regions. 
With the new pathways, the agents from their corresponding regions can still travel and reach the other regions, overcoming the obstacles due to disruptions. 


\begin{figure}[H]
	\centering
\begin{tikzpicture}[scale=1.1]

\draw[thick] (-4.5,0) -- (4.5,0);

\foreach \x in {-4.2,-3.2,-1.2,-0.4,2.8,3.3}
{
    \fill (\x,0) circle (2.2pt);
}

\node[below=3pt] at (-2.0,0.7) {\Large $a$};
\node[below=5pt] at (0,0) {\large $o$};
\node[below=1.5pt] at (1.2,0) {\large $o+L$};
\node[below=3pt] at (2.2,0.8) {\Large $b$};

\draw[green!60!black, very thick]
    (-2.0,0.05) .. controls (-0.7,0.75) and (0.9,0.75) .. (2.2,0.05);

\draw[thick, decorate,
      decoration={coil, aspect=0.35, segment length=5pt, amplitude=8pt}]
      (0.1,0.05) -- (1.2,0.05);

\draw[decorate,
      decoration={brace, amplitude=8pt, mirror},
      gray, thick]
      (-4.5,-0.55) -- (-0.1,-0.55);

\node at (-2.35,-1.05) {Region $A$};

\draw[decorate,
      decoration={brace, amplitude=8pt, mirror},
      gray, thick]
      (1.2,-0.55) -- (4.5,-0.55);

\node at (2.85,-1.05) {Region $B$};

\end{tikzpicture}
\caption{An obstacle ranging from $o$ to $o+L$ disconnects the agents (denoted as solid points) in two regions $A$ and $B$, and a new pathway $(a,b)$ connects them. 
    }
	\label{fig:obs}
\end{figure}
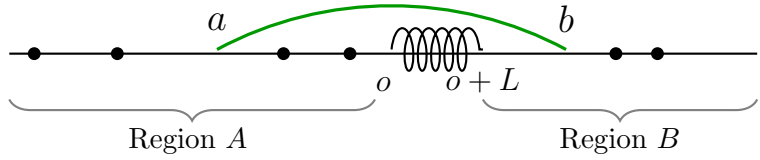

Figure~\ref{fig:obs} provides a motivating example of a structurally deficient bridge that is suspended and therefore becomes impassable. 
We model the original bridge as a one-dimensional infrastructure, where the suspended segment is represented by the obstacle interval $[o,o+L]$. 
Before the suspension, the bridge connects Region~$A$ and Region~$B$; after the suspension, the obstacle disconnects the two regions and prevents agents from traveling directly between them. 
Nevertheless, agents in each region may still need to access the other region for work, school, services, or other daily activities. 
This motivates the construction of an alternative pathway or replacement bridge, represented by the green curve with endpoints $a$ and $b$ in Figure~\ref{fig:obs}, to restore accessibility between Region~$A$ and Region~$B$.

Existing optimization literature has considered building optimal pathways  between two disconnected regions, aiming to minimize the maximum distance between any two points from the two regions (see, e.g., \cite{leizhen1999computing,Kim1998Linear,kim2001computing,tan2000optimal,tan2002finding}). While these studies designed polynomial algorithms for building optimal bridges between different types of convex polygons (more details in related work), 
there are two main assumptions that make the current optimization literature not ideal for capturing real-world situations under disruptions involving agents. 
First, existing literature assumes that the agents are located in all of the points in the regions. 
However, in many real-world situations, agents’ locations consist only of a subset of discrete points in the regions. 
Second, existing literature assumes that each agent's location is public information. 
However, agent locations might not be known in advance and require elicitation \cite{DBLP:books/cu/NRTV2007,procaccia2013approximate}. 
Therefore, our goal is to build optimal pathways to account for agents’ locations to connect them to the respective regions. 

\subsection{Our Contribution} 

We initiate the mechanism design study of building (approximately) optimal pathways  between two regions disconnected by obstacles under disruptions to connect agents from their respective regions to other regions. 
In out setting, a line segment (denoted by an interval $[0,1]$) connecting two regions is separated by an obstacle. The left endpoint of the obstacle is at $o\in (0, 1)$, and the right endpoint is at $o+L$, where $L$ is the distance of the obstacle. (see Figure \ref{fig:obs}.\footnote{The line space has been extensively studied in mechanism design of facility location problems for modeling  geographic regions and other real-world non-geographic  situations \cite{chan2021survey,procaccia2013approximate}.} Agents in the regions are denoted by sets $N_1$ and $N_2$, depending on whether their locations are points on the left-hand side or the right-hand side of the obstacle (i.e., $x_i \in  [0, o)$ or $x_i \in (o+L, 1]$ for any agent $i$ in $N_1$ or $N_2$).

We aim to design mechanisms to elicit agent locations and build a pathway/edge  $(a,b)$ that connects the two disconnected regions, where the left endpoint $a$ is in $[0,o)$ and the right endpoint $b$ is in $(o+L,1]$. 
Given an edge $(a,b)$, the cost of an agent at $x_i\in [0, o)$ is $|x_i-a|+k(b-a)+1-b$ with $k$ being a non-negative multiplication factor, that is, the distance from their location to the (farthest) endpoint on the other region, passing through edge $(a,b)$. 
The cost of an agent at $x_i\in (o+L,1]$ is defined similarly as $|x_i-b|+k(b-a)+a$. 
We consider adding edges that minimize two different objectives:  the social cost (i.e., the total  cost of all agents) and the maximum cost (i.e., the maximum cost  among all agent costs). When $k\ge 1$, a mechanism that returns an optimal solution $(o-\epsilon,o+L+\epsilon)$ is group strategyproof for $\epsilon\rightarrow 0$. Therefore, we only need to focus on the situation when $k\in[0,1)$.\footnote{{In various situations, the social planner can determine the value of $k$ appropriately. 
For instance, $k >> 1$ can be set for constructing a temporary detour. 
When creating a new road or bridge (to replace
an older one), the social planner can set  
$k < 1$ 
by making it wider or having higher speed limits. }
}

We first provide a characterization of all strategyproof and anonymous mechanisms as two-dimensional generalized median mechanisms by showing that the agent preferences over the locations on where to build the pathway are two-dimensional single-peaked. {The single-peakedness means that agents have preferences over a set of options (i.e., pathway locations in our setting) that can be ordered, so that each agent has the most preferred option (called the peak) and their preference for other options decreases as they move away from this peak. See more details in Section \ref{sec:cha}.}

For the social cost, we derive an optimal solution on where to build a pathway and show that the mechanism that returns the optimal solution is group strategyproof. 
For the maximum cost, we show that there is a unique optimal solution on where to build a pathway and the optimal solution is not strategyproof. We show that a deterministic group strategyproof mechanism,  \textsc{TwoExtreme}, that simply connects two agent locations nearest to the obstacle has an approximation ratio of $\frac{2-2(1-k)L}{1+k-(1-k)L}$. Specially, when $L=0$, it is $\frac{2}{1+k}$. 
On the other hand, we show a lower bound of deterministic mechanisms. Specially, when $L=0$, no deterministic strategyproof mechanism has an approximation ratio less than $\frac{2}{1+\sqrt{k}}$. 

In Section~\ref{sec:special}, we study the special case $L=0$ of a singleton obstacle for maximum cost. We provide an improved deterministic mechanism, \textsc{TwoExtremeRestrict},  by not allowing the endpoints of the edge to be too close to the obstacle (see Theorem \ref{thm:round2}). Besides, we improve the lower bound from $\frac{2}{1+\sqrt{k}}$ to the experimental lower bound in Figure \ref{fig:res}. 
Moreover, we design a randomized group strategyproof  mechanism   \textsc{RandMaxCost} that has an approximation ratio of $\max\left(\frac{4-2k}{3-k}, \frac{1+k}{1+k^2}\right)$.  
We also provide a lower bound $\frac{6+6k}{5+7k}$ for any randomized strategyproof mechanisms. 
See Figure \ref{fig:res} for an illustration of the above bounds when $L=0$.

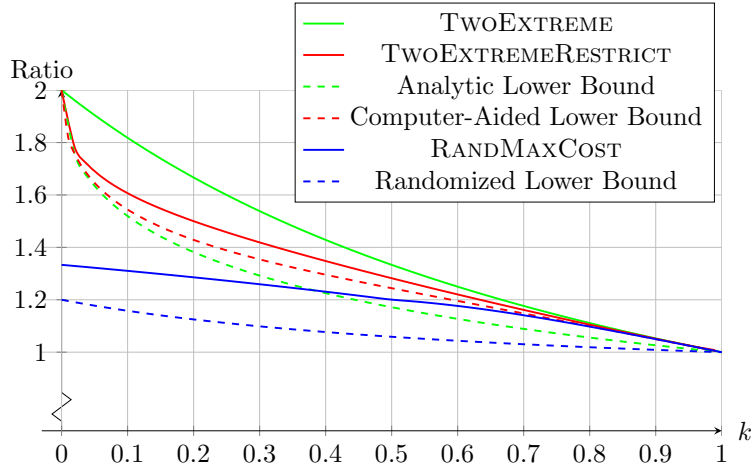
\begin{figure}[h]
	\centering
\begin{tikzpicture}[scale=0.90]
\begin{axis}[axis x line=middle,
             axis y line=middle,
             grid,
             ylabel=Ratio,
             xlabel=$k$,
             xlabel style={at={(1.01,0)}, anchor=west},
             ylabel style={at={(0,1.01)}, anchor=south},
             legend style={at={(0,0)},anchor=south east,at={(axis description cs:0.99,0.67)}},
             axis y discontinuity=crunch,
             extra x ticks={0},
             extra x tick labels={$0$},
             xmin=-0.03,ymin=0.7,
             xtickmin=0,ytickmin=1,
             width=10cm,height=5cm,
             scale only axis,
             ]
    
     \addplot[smooth,thick,samples=50,green,domain=0:1]{2/(1+x)};
    \addlegendentry{ \textsc{TwoExtreme}}
    \addplot[smooth,thick,samples=50,red,domain=0:0.99]{2*max(
        (1-((1+x^2-sqrt(x^4-x^3+3*x^2+x))/(1+x)))/(1+x-((1+x^2-sqrt(x^4-x^3+3*x^2+x))/(1+x))),
        (x*(2*((1+x^2-sqrt(x^4-x^3+3*x^2+x))/(1-x^2))-((1+x^2-sqrt(x^4-x^3+3*x^2+x))/(1-x^2))^2)+1-((1+x^2-sqrt(x^4-x^3+3*x^2+x))/(1-x^2))^2)/(2-2*((1+x^2-sqrt(x^4-x^3+3*x^2+x))/(1+x))), 
        (1+2*((1+x^2-sqrt(x^4-x^3+3*x^2+x))/(1-x^2))*x)/(2-((1+x^2-sqrt(x^4-x^3+3*x^2+x))/(1+x))),
        ((1+x^2-sqrt(x^4-x^3+3*x^2+x))/(1-x^2))
    )};
    \addlegendentry{ \textsc{TwoExtremeRestrict}}
    \addplot[dashed,smooth,thick,samples=50,green,domain=0:1]{max(2/(1+sqrt(x)),1};
    \addlegendentry{Analytic Lower Bound}
    \addplot[dashed, red, thick, smooth] table[x=k, y=Lk, col sep=comma] {lower_bound_data.csv};
    \addlegendentry{Computer-Aided Lower Bound}
    \addplot[smooth,thick,samples=50,blue,domain=0:1]{1+max(((max((x^2+x)/(1+x^2), (1+x)/(3-x)))*(1-x))/(1+x), (1-(max((x^2+x)/(1+x^2), (1+x)/(3-x))))*x)};
    \addlegendentry{\textsc{RandMaxCost}}
    \addplot[dashed,smooth,thick,blue,samples=50,domain=0:1]{(6+6*x)/(5+7*x)};
    \addlegendentry{Randomized Lower Bound}
\end{axis}
\end{tikzpicture}
	\caption{An illustration  of upper and lower bounds for the maximum cost when  $k \in [0,1)$ and $L=0$. The upper bounds of our mechanisms are depicted in solid lines, and the lower bounds are in dashed lines. 
 }
		\label{fig:res}
\end{figure}

All our results apply to the setting where the obstacle is a closed subinterval of $[0,1]$, since such an interval can be reduced to a point as in our model.

\paragraph{Organization.} We introduce the model in Section~\ref{sec:model}. A characterization of all strategyproof mechanisms is provided in Section~\ref{sec:cha}. The social cost and the maximum cost objectives are analyzed in Sections~\ref{sec:soc} and~\ref{sec:max}, respectively.  Section~\ref{sec:special} is devoted to the special case of $L=0$.

\subsection{Related Works} 
While no existing mechanism design literature considers our setting, we discuss the most related optimization studies on building optimal bridges connecting two regions and adding edges to discrete networks to improve network parameters. 
We also discuss the related works in the approximate mechanism design without money paradigm. 

\paragraph{Bridge-building.} 
Existing optimization literature has considered the problem of building an optimal bridge to connect two disconnected regions. 
Cai et al.
\cite{leizhen1999computing} introduced the problem of adding a line segment to connect two disjoint convex polygonal regions in a plane, such that the length of the longest path from a point in one polygon, passing through the bridge, to a point in another region is minimized.
They proposed an $O(n^2 \log n)$-time algorithm, where $n$ is the maximum number of extreme points of the polygons. 
Later, \cite{bhattacharya2001computing} proposed a linear-time algorithm that improves the $O(n^2 \log n)$-time algorithm in \cite{leizhen1999computing}.  
Tan \cite{tan2000optimal} independently presented an alternative linear-time algorithm for the above setting and further generalized it to an $O(n^2)$-time algorithm for bridging two convex polyhedra in space.
 \cite{kim2001computing} provided algorithms to find an optimal bridge between two convex polygons, two simple non-convex polygons, and one convex and one simple non-convex polygons in $O(n)$,  $O(n^2)$, and $O(n\log n)$, respectively. 
Later, Tan \cite{tan2002finding} provided an $O(n\log^3 n)$-time algorithm for the settings of two simple non-convex polygons. 
Kim et al. \cite{Kim1998Linear} proposed a linear-time algorithm to compute an optimal bridge between two parallel lines separated by an obstacle to minimize the length of the longest path connecting two points on the lines. 
However, all of the above-mentioned works focus on all points in the regions. Our work focuses on a finite subset of points, which are the agents' locations, and the mechanism perspective in which agents' locations are private. 



\paragraph{Edge addition on networks.} 
Existing optimization studies have examined adding edges to discrete networks (with nodes and edges) to minimize the diameter or average shortest distances between pairs of nodes of a network (see, e.g., \cite{demaine2010minimizing,meyerson2009minimizing,papagelis2011suggesting,perumal2013minimizing}. 
However, all these optimization studies on discrete networks do not consider disconnected regions that are continuous and assume agents occupy all nodes/vertices of the network. Moreover, they do not consider the mechanism design perspective. 


\paragraph{Mechanism design.} 
Our considered mechanism design setting is within the paradigm of approximate mechanism design without money, initialized by Procaccia and Tennenholtz \cite{procaccia2013approximate} who used facility location problems (FLPs) as case studies. This paradigm investigates the design of approximately optimal strategyproof mechanisms through the lens of the approximation ratio. 
In a typical setting of FLPs, the agents report their private locations on the real line to a mechanism. The mechanism determines the locations for building facilities to minimize some  objectives involving the costs of  agents, where the cost of each agent is their distance to the facilities. 
Following their work, variations of FLPs have been introduced and studied (see, e.g., \cite{dokow2012mechanism,feldman2013strategyproof,filos2021approximate,lin2020nearly,mei2019facility,meir2019strategyproof}).
{We note that the case $k=0$ of our setting is equivalent to a 2-FLP problem where each agent $i$ has two locations, $x_i$ and the endpoint of the other region (0 or 1), whose cost is the total distance from their two locations to the two facilities (which are now represented as a pathway with $k=0$). }
We refer readers to a survey on  mechanism design for FLPs  \cite{chan2021survey}. 
The most relevant mechanism design work to ours is the work of \cite{Chan023} in which they considered modifying the structure of regions by adding a shuttle or road to improve the distances of the agents to a prelocated facility. In contrast, they do not consider two regions separated by an obstacle.

\section{Model}\label{sec:model}

Let $N=\{1,\ldots,n\}$ be the set of agents located in an interval $[0,1]$. The location profile of agents is denoted as $\mathbf x=(x_1,\ldots,x_n)$. There is an obstacle represented by $[o,o+L]$, where $o\in (0,1)$ is the left endpoint, and $L\in [0, 1-o)$ is the length. Provided that no agent is located within this obstacle, it partitions the agents into $N=(N_1,N_2)$ according to their regions, where $N_1=\{i\in N~|~x_i<o\}$ is the set of agents on the left region, and $N_2=\{i\in N~|~x_i>o+L\}$ is the set of agents on the right region. 
The agents on one region are required to access the other region. 
Due to the obstacle, the agents cannot pass through it and reach the other region directly. 
Hence, we want to build a new edge $(a,b)$ that connects the two regions with $a\in[0,o)$ and $b\in(o+L,1]$. 
The length of the edge is $k(b-a)$, where $k$ is a positive constant. 

A deterministic mechanism $f:\mathbb R^n\rightarrow \mathbb R^2$ is a function that takes the agent location profile $\mathbf x$ as input and returns an edge $f(\mathbf x)$. Given an edge $f(\mathbf x)=(a,b)$, the \emph{cost} of each agent $i\in N_1$ on the left region is the distance to  the right endpoint $1$ through the edge, 
$$cost(a,b,x_i)=|x_i-a|+k(b-a)+(1-b).$$
Similarly, the \emph{cost} of each agent $i\in N_2$ on the right region is the distance to the left endpoint $0$ through the edge, 
$$cost(a,b,x_i)=|x_i-b|+k(b-a)+a.$$
A randomized mechanism is a function $f$ from $\mathbb R^n$
to probability distributions over $\mathbb R^2$.  If $f(\mathbf x) = P$ is a probability distribution, the cost of agent $i\in N$ is
defined as the expected cost  $cost(P,x_i)=\mathbb E_{(a,b)\sim P}[cost(a,b,x_i)]$. 

A mechanism $f$ is \emph{strategyproof} if no agent can decrease their cost by misreporting the location within their region.
Formally, $f$ is strategyproof if for any $i\in N,\mathbf x$ and $x_i'$ with $x_i,x_i'$ on the same region,
$cost(f(x_i,\mathbf x_{-i}),x_i)\le cost(f(x_i',\mathbf x_{-i}),x_i)$,
where $\mathbf x_{-i}$ is the location profile of the agents in $N\setminus\{i\}$.
Further, $f$ is called \emph{group strategyproof} if no group of agents can misreport simultaneously so that all agents in the group are better off. That is, for any $S\subseteq N,\mathbf x,\mathbf x_S'$, there exists an agent $i\in S$ such that 
$cost(f(\mathbf x),x_i)\le cost(f(\mathbf x_S',\mathbf x_{-S}),x_i)$. 
A mechanism $f$ is \emph{anonymous} if the outcomes are invariant under
permutation of agents, i.e., $f(x_1,\ldots,x_n)=f(x_{\pi(1)},\ldots,x_{\pi(n)})$ for every profile $\mathbf x$ and every permutation of agents $\pi:N\rightarrow N$. Since a non-anonymous mechanism is based on the identity
of the agents and is much less interesting, we focus only on anonymous mechanisms.

Our goal is to design (group) strategyproof mechanisms with good performance guarantees under two objectives: minimizing the social cost and minimizing the maximum cost. The social cost with respect to an edge $(a,b)$ and location profile $\mathbf x$ is 
$$SC(a,b,\mathbf x)=\sum_{i\in N}cost(a,b,x_i).$$
The maximum cost with respect to an edge $(a,b)$ and location profile $\mathbf x$ is 
$$MC(a,b,\mathbf x)=\max_{i\in N}cost(a,b,x_i).$$
A mechanism $f$ is $\alpha$-approximation ($\alpha\ge 1$) for the objective function $\Delta\in\{SC,MC\}$ if 
$\Delta(f(\mathbf x),\mathbf x)\le \alpha\cdot \min_{(a,b)\in\mathbb R^2}\Delta(a,b,\mathbf x)$ for all location profiles $\mathbf x\in\mathbb R^n$.

\medskip
We remark that when the constant coefficient is $k\ge 1$, there is a trivial solution $(o-\epsilon,o+L+\epsilon)$ for some fixed value $\epsilon>0$. As $k\ge 1$, every agent wants the edge to be as short as possible, that is, $\epsilon$ approaches $0$. Then, a mechanism that returns the fixed solution $(o-\epsilon,o+L+\epsilon)$ is clearly group strategyproof and (almost) optimal for both objectives when $\epsilon\rightarrow 0$. Therefore, 
in the remainder of this paper, we assume that $k\in[0,1)$. 
{
Because it is natural for each agent to only misreport locations within their own region, our model assumes that agents cannot misreport their locations in other regions. However, it is worth noting that all of our mechanisms (Mechanism \ref{mec:socopt}-\ref{alg:ran3}) are still strategyproof and retain the same approximation even without this assumption.  }

\section{Characterizing  Strategyproof Mechanisms}\label{sec:cha}

In this section, we show that the preference profile of agents is multi-dimensional single-peaked, and the generalized median mechanisms compose the class of all anonymous strategyproof mechanisms.

We start with some necessary definitions. Let $D$
be a set of possible outcomes. A one-dimensional axis $A$ on $D$ is any strict ordering
$<_A$ of the outcomes in $D$. A multi-dimensional axis $A^m=\langle A_1,\ldots,A_m\rangle$ on $D$ is a collection of $m$ distinct axes, each being a one-dimensional axis on $D$. 
\begin{definition}[\cite{barbera1993generalized}]\label{def:mul}
     Let $A^m$ be an $m$-dimensional axis on the set $D$ of possible outcomes. An agent $i$'s preference $\succeq_i$ is \emph{$m$-dimensional single-peaked with respect to $A^m$} if:
(1) there is a single most-preferred outcome (peak) $p_i\in D$, and 
    (2) for any two outcomes $\alpha,\beta\in D$, $\alpha\succeq_i\beta$ whenever $\beta<_{A_t}\alpha<_{A_t} p_i$ or $p_i<_{A_t}\alpha <_{A_t}\beta$ for all axes $A_t$, $t=1,\ldots,m$. 
\end{definition}

Then, a preference profile is called \emph{$m$-dimensional single-peaked} if there exists an $m$-dimensional axis
$A^m$ such that every agent preference is $m$-dimensional single-peaked with respect to $A^m$. While the preference profile in our problem is not one-dimensional single-peaked (see the detailed explanation for $L=0$ in Appendix \ref{app:a}), it is indeed two-dimensional single-peaked. 

\begin{theorem}\label{thm:sing}
    For any instance of our problem, the preference profile of agents is 2-dimensional single-peaked.
\end{theorem}
\begin{proof}
    In our problem, an outcome is a shortcut edge $(a,b)$ with $a\in[0,o)$ and $b\in (o+L,1]$, and it uniquely corresponds to a point $(a,b)$ in the 2-dimensional $xy$-coordinate system. Thus, the set of all possible outcomes can be represented by a set $D=\{(x,y)\in \mathbb R^2~|~0\le x<o\le o+L<y\le 1\}$ in a plane, as shown in Figure \ref{fig:single}. 

\begin{figure}[H]
    \centering
    \includegraphics[width=5cm]{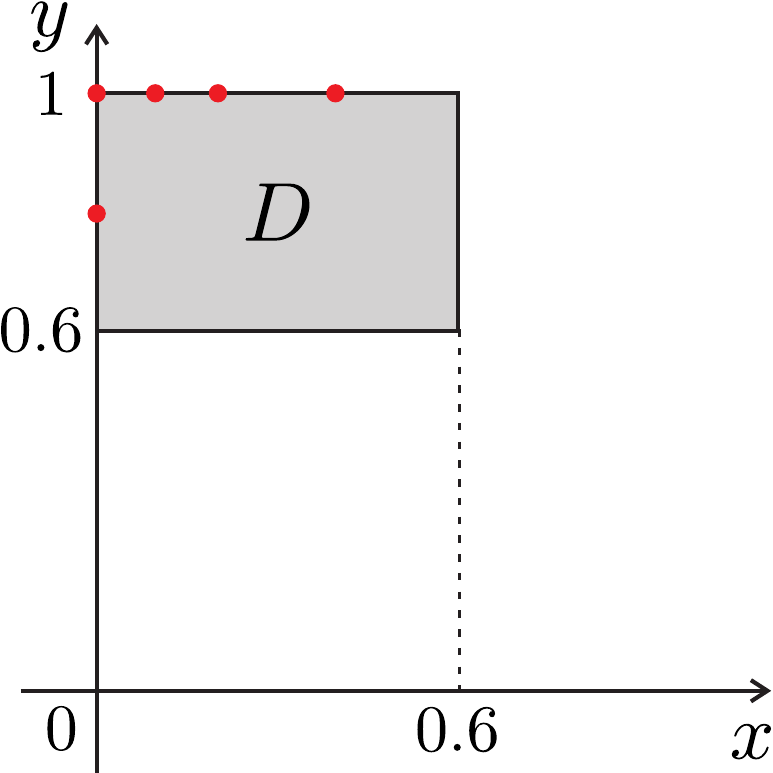}
    \caption{An illustration of the set $D$ when the obstacle is at $o=0.6$ with length $L=0$. For location profile $\mathbf x=(0.1,0.2,0.4,0.8,1)$, the agent peaks are $(0.1,1),(0.2,1),(0.4,1),(0,0.8),(0,1)$, denoted as red points. 
    }
    \label{fig:single}
\end{figure}

    Now we show that every agent is $2$-dimensional single-peaked with respect to the collection of $x$-axis and $y$-axis. For any agent $i\in N_1$, the single most-preferred outcome (peak) is $p_i=(x_i,1)\in D$. For any two outcomes (points) $(x,y),(x',y')\in D$, if they satisfy (1) $x'<x<x_i$ or $x_i<x<x'$, and (2)  $y'<y<1$, then we have 
    \begin{align*}
       & c(x,y,x_i)=|x_i-x|+k(y-x)+1-y
        < |x_i-x'|+k(y'-x')+1-y'  =c(x',y',x_i),
    \end{align*}
    implying that agent $i$ has a smaller cost under $(x,y)$ than that under $(x',y')$ and the agent prefers $(x,y)$. On the other hand, for any agent $j\in N_2$, the single most-preferred outcome is $p_j=(0,x_j)\in D$. For any two outcomes  $(x,y),(x',y')\in D$, if they satisfy (1) $y'<y<x_j$ or $x_j<y<y'$, and (2)  $0<x<x'$, similarly,  it is easy to see that agent $j$ has a smaller cost under $(x,y)$ than that under $(x',y')$, and thus the agent  prefers $(x,y)$. Hence, both conditions in Definition \ref{def:mul} are satisfied for every agent, and the preference profile is 2-dimensional single-peaked.
\end{proof}

For (one-dimensional) single-peaked preferences in a one-dimensional space, a mechanism is a  \emph{generalized
median mechanism} if there exists $n+1$ constants $b_1,\ldots,b_{n+1}\in\mathbb R\cup\{-\infty,+\infty\}$ such that the outcome is $\text{med}(p_1,\ldots,p_n,b_1,\ldots,b_{n+1})$, where $\text{med}$ is the median function, and $p_1,\ldots,p_n$ is the most-preferred outcome of the $n$ agents. 
For $m$-dimensional single-peaked preferences in an $m$-dimensional space, an \emph{$m$-dimensional generalized median mechanism} can be decomposed into $m$ independent one-dimensional generalized median mechanisms, with the $t$-th mechanism determining the coordinate of the outcome on the $t$-th dimension, for all $t=1,\ldots,m$ (see, e.g., \cite{sui2015mechanism}).   Barber{\`a} \emph{et al.} \cite{barbera1993generalized} provide a characterization result: 
    a mechanism for multi-dimensional single-peaked preferences in a multi-dimensional space is strategyproof and anonymous
if and only if it is a multi-dimensional generalized median mechanism. This characterization applies to our problem by Theorem \ref{thm:sing}.

\begin{corollary}\label{cor:cha}
    A mechanism for our problem is strategyproof and anonymous
if and only if it is a $2$-dimensional generalized median mechanism. 
\end{corollary}

\begin{corollary}\label{cor:lower}
    Let $f$ be a deterministic anonymous strategyproof mechanism. Consider a profile $\mathbf{x}$ for which $f$ outputs facility locations $(a, b) = (a_0, b_0)$. If an agent $i \in N_1$ moves closer to or away from $a_0$ (or $i \in N_2$ moves closer to or away from $b_0$) without crossing it, then the output $(a, b)$ remains unchanged. 
\end{corollary}

\section{Social Cost}\label{sec:soc}

In this section, we consider the social cost. We start with some necessary notations. 
For each point $x\in[0,o)$, define $N_{1l}(x)=\{i\in N_1|x_i\le x\}$ to be the set of agents in $N_1$ on the left of $x$, and  $N_{1r}(x)=\{i\in N_1|x_i> x\}$ to be the set of agents in $N_1$ on the right of $x$. For each point $x\in (o,1]$, define $N_{2l}(x)=\{i\in N_2|x_i< x\}$ to be the set of agents in $N_2$ on the left of $x$, and  $N_{2r}(x)=\{i\in N_2|x_i\ge  x\}$ to be the set of agents in $N_2$ on the right of $x$. 

We consider a mechanism that determines the location of the two endpoints separately. First, the mechanism fixes $b$ and moves $a$ from  $0$ towards the obstacle as long as the social cost is decreasing. Second, it moves $b$ from  $1$ towards the obstacle as long as the social cost is decreasing. 

\begin{mechanism}[\textsc{OptSocCost}]\label{mec:socopt}
   Given location profile $\mathbf x$, let $X_L$ be the set of all points $x\in[0,o)$ that satisfy  \begin{equation*}\label{eq:mm}
       |N_{1l}(x)\cup N_2|\cdot (1-k)< |N_{1r}(x)|\cdot (1+k),
   \end{equation*}
   and let $a^*=\sup X_L$ if $X_L$ is non-empty and $a^*=0$ otherwise.
Let $X_R$ be the set of all points $x\in(o,1]$ that satisfy  $$|N_{2r}(x)\cup N_1|\cdot (1-k)< |N_{2l}(x)|\cdot (1+k),$$
and let $b^*=\inf X_R$ if $X_R$ is non-empty and $b^*=1$ otherwise.
Return $(a^*,b^*)$.
\end{mechanism}

Note that for every $x\in [0,o)$ and fixed $b$, if we move the left endpoint $a$ from $x$ to $x+\epsilon$ for some sufficiently small value $\epsilon>0$, the cost of the agents in $N_{1l}(x)\cup N_2$ increases by $(1-k)\epsilon$, and the cost of the agents in $N_{1r}(x)$ decreases by $(1+k)\epsilon$. Thus, if we move $a$ from  $0$ towards the obstacle as long as the social cost is decreasing, it approaches the supremum of $X_L$. Symmetrically, we move $b$ towards the obstacle and it approaches the infimum of $X_R$. 
Since the supremum of $X_L$ and  the infimum of $X_R$ can only be the agent locations, the mechanism is in polynomial time. 

{When $k=0$, \textsc{OptSocCost} is exactly the median mechanism that selects the $x$-coordinate (resp. $y$-coordinate) of the outcome to be the median of the $x$-coordinates (resp. $y$-coordinates) of $n$ agent peaks. 
}



\begin{theorem}
    Mechanism \ref{mec:socopt} is group strategyproof and optimal for the social cost.
\end{theorem}
\begin{proof}
For the group strategyproofness, we consider a group of agents $S\subseteq N_1\cup N_2$. Let $f(\mathbf x)=(a,b)$ be the outcome when all agents report true locations, and $f(\mathbf x_{S}',\mathbf x_{-S})=(a',b')$ be the outcome when the agents in $S$ misreport $\mathbf x_{S}'$. Assume w.l.o.g. that $|a-a'|\ge |b-b'|$. We show that at least one agent in the group cannot gain by misreporting. 

If $a'< a$, the agents in $N_{1r}(a)$ cannot gain and they are not in the group. Since the agents in $N_2$ cannot change the location of $a$, by the definition of the mechanism,  an agent located at $a$ must be in the group and misreport a location to the left of $a$.  The cost of this agent decreases by at most $(1-k)|b'-b|-(1+k)(a-a')\le 0$, indicating that this agent can never gain. 

If $a'> a$, by the mechanism, there  exists at least one agent in $N_{1l}(a)$ and in group $S$ who misreport the location to the right of $a$. However, the cost of such an agent in $N_{1l}(a)\cap S$ will decrease by at most $(1-k)|b'-b|-(1-k)(a'-a)\le 0$, and the agent cannot gain by misreporting. 

For the optimality, let $(a^*,b^*)$ be the solution returned by the mechanism, and let $(a,b)$ be an optimal solution.  Assume w.l.o.g. that $|a-a^*|\ge |b-b^*|$.
If $a<a^*$, we consider a new solution $(a+\epsilon,b)$, where $\epsilon>0$ is a sufficiently small number so that there is no agent in the interval $(a,a+\epsilon)$. Compared with solution $(a,b)$, each agent in $N_{1l}(a)\cup N_2$ increases the cost by $(1-k)\epsilon$, and each agent in $N_{1r}(a)$ decreases the cost by $(1+k)\epsilon$. By the definition of $a^*$ and the fact that $a<a^*$, we have 
$|N_{1l}(a)\cup N_2|\cdot (1-k)< |N_{1r}(a)|\cdot (1+k)$, indicating that the social cost of $(a,b)$ is larger than that of $(a+\epsilon,b)$, which contradicts the optimality. 

If $a>a^*$, we change the solution $(a,b)$ to $(a^*,b)$.  Each agent in $N_{1l}(a^*)\cup N_2$ decreases the cost by $(1-k)(a-a^*)$, and the increase of cost for every agent in $N_{1r}(a^*)$ is at most $(1+k)(a-a^*)$. By the definition of $a^*=\sup X_L$ and the fact that $X_L$ is an open set, we have 
$|N_{1l}(a^*)\cup N_2|\cdot (1-k)\ge  |N_{1r}(a^*)|\cdot (1+k)$.
It indicates that the social cost of $(a^*,b)$ is no more than that of $(a,b)$, and $(a^*,b)$ is also optimal. 

Then we change the solution $(a^*,b)$ to $(a^*,b^*)$. By a symmetric analysis, we know that $b>b^*$ is impossible, and if $b\le b^*$,   the social cost of $(a^*,b^*)$ is no more than that of $(a^*,b)$. Hence, $(a^*,b^*)$ is optimal. 
\end{proof}


\section{Maximum Cost}\label{sec:max}

In this section, we consider the maximum cost. We first characterize the unique optimal solution and show that the mechanism that returns the optimal solution is not strategyproof. 
We then study deterministic and randomized mechanisms. 

\subsection{The Optimal Maximum Cost}\label{sec:maxopt}

Given location profile $\mathbf x$, let $x_l=\min\{x_i|i\in N_1\}$ and $x_r=\max\{x_i|i\in N_1\}$ be the two extreme agent locations to the left, and let $y_l=\min\{x_j|j\in N_2\}$ and $y_r=\max\{x_j|j\in N_2\}$ be the two extreme agent locations to the right.

\begin{lemma}\label{lem:opt1}
    In any optimal solution $(a,b)$ for maximum cost, we have $a\le \frac{x_l+x_r}{2}$ and $b\ge \frac{y_l+y_r}{2}$.
\end{lemma}
\begin{proof}
    Suppose that $a> \frac{x_l+x_r}{2}$, which indicates that $cost(a,b,x_r)< cost(a,b,x_l)$. Clearly, the maximum cost is attained by $x_l$ or $y_l$ or $y_r$.  We consider a solution $(a-\epsilon,b)$ that moves the left endpoint to a sufficiently small positive value $\epsilon<a-\frac{x_l+x_r}{2}$. Then the cost of the agents at $x_l,y_l,y_r$ decreases, and thus the maximum cost decreases, which contradicts the optimality. Therefore, it must be $a\le \frac{x_l+x_r}{2}$. By a symmetric analysis, we can prove 
    $b\ge \frac{y_l+y_r}{2}$.
\end{proof}

\begin{lemma}\label{lem:opt2}
    In any optimal solution $(a,b)$ for maximum cost, either $a= \frac{x_l+x_r}{2}$, or $b= \frac{y_l+y_r}{2}$.
\end{lemma}
\begin{proof}
    Suppose that $a\neq \frac{x_l+x_r}{2}$ and $b\neq \frac{y_l+y_r}{2}$. By Lemma \ref{lem:opt1}, we have $a< \frac{x_l+x_r}{2}$ and $b> \frac{y_l+y_r}{2}$. It is easy to see that the maximum cost is attained by either $x_r$ or $y_l$.
  We consider a solution $(a+\epsilon,b-\epsilon)$ that moves the two endpoints towards the obstacle by a sufficiently small value $\epsilon>0$.  
  The cost of both agents at $x_r$ and $y_l$ decreases by $(1+k)\epsilon-(1-k)\epsilon=2k\epsilon$, and thus the maximum cost decreases, which contradicts the optimality. Therefore, it must be $a= \frac{x_l+x_r}{2}$, or $b= \frac{y_l+y_r}{2}$, or both hold. 
\end{proof}

Now we describe the algorithm \textsc{OptMaxCost} to derive the optimal maximum cost. 
\begin{mechanism}
    (\textsc{OptMaxCost}) Given location profile $\mathbf x$,  
 if $1-y_r\ge x_l$, define $a^*= \frac{x_l+x_r}{2}$ and $b^*= \frac{y_l-x_l}{2}+\frac12$.
     If $1-y_r<x_l$, define $a ^*= \frac{x_r-y_r}{2}+\frac12$ and $b^*=\frac{y_l+y_r}{2}$. Return $(a^*,b^*)$.
\end{mechanism}



Notice that when $1-y_r=x_l$, the mechanism simply returns the two midpoints $(\frac{x_l+x_r}{2},\frac{y_l+y_r}{2})$.

\begin{theorem}\label{thm:opt}
   \textsc{OptMaxCost} returns the unique optimal solution for maximum cost. In addition, the optimal maximum cost is attained by both the agents located at $x_r$ and $y_l$. 
\end{theorem}
\begin{proof}
We focus only on the case when $1-y_r\ge x_l$, as the other case is symmetric. The mechanism returns $(a^*,b^*)=(\frac{x_l+x_r}{2},\frac{y_l-x_l}{2}+\frac12$).  The maximum cost under $(a^*,b^*)$ is attained by the agents at $x_l$, $x_r$ and $y_l$, that is,  
\begin{align*}
cost(a^*,b^*,x_l)&=\frac{x_r-x_l}{2}+k(b^*-a^*)+\frac12-\frac{y_l-x_l}{2}    \\
&= \frac12-\frac{y_l+x_l}{2}+k(b^*-a^*)+\frac{x_l+x_r}{2} = cost(a^*,b^*,y_l),
\end{align*}
which is no less than $cost(a^*,b^*,y_r)$ since $b^*\ge \frac{y_l+y_r}{2}$.

Suppose that $(a,b)$ is an optimal solution. By Lemma \ref{lem:opt2}, we have  either $a= \frac{x_l+x_r}{2}$ or $b= \frac{y_l+y_r}{2}$. When $a= \frac{x_l+x_r}{2}$, if $b<b^*$, then the agents at $x_l$ and $x_r$ have a larger cost in $(a,b)$ than that in $(a^*,b^*)$, implying that $(a,b)$ is not optimal for maximum cost. If $b>b^*$, then the agent at $y_l$ has a larger cost in $(a,b)$ than that in $(a^*,b^*)$, and thus $(a,b)$ is not optimal. Therefore, $(a^*,b^*)$ is the unique optimal solution. 

When $b= \frac{y_l+y_r}{2}$, we have $a\le \frac{x_l+x_r}{2}$ by Lemma \ref{lem:opt1}, and the maximum cost induced by $(a,b)$ is at least
\begin{align*}
    cost(a,b,x_r)&=x_r-a+k(\frac{y_l+y_r}{2}-a)+1-\frac{y_l+y_r}{2}\\
    &\ge x_r-(1+k)\frac{x_l+x_r}{2}+1-(1-k)\frac{y_l+y_r}{2}\ge cost(a^*,b^*,y_l),
\end{align*}
where the equation $cost(a,b,x_r)=cost(a^*,b^*,y_l)$ holds only if $1-y_r= x_l$ and $a=\frac{x_l+x_r}{2}$, that is, $(a,b)=(a^*,b^*)$. Therefore, $(a^*,b^*)$ is the unique optimal solution. 
\end{proof}

However, \textsc{OptMaxCost} is not strategyproof. Consider a location profile $\mathbf x=(0,0.2,0.8,1)$, and the obstacle is at $0.5$ with $L=0$.  The mechanism returns the two midpoints $(\frac{x_l+x_r}{2},\frac{y_l+y_r}{2})=(0.1,0.9)$, and the cost of the agent at $0.2$ is $0.2+0.8k$. If this agent misreports the location as $0.4$, then the outcome of the mechanism becomes $(0.2,0.9)$, and the cost of this agent with true location $0.2$ decreases to $0.1+0.7k<0.2+0.8k$.


\subsection{Deterministic Mechanisms}\label{sec:maxdet}

We consider designing deterministic strategyproof mechanisms with good performance guarantees. {We first present a simple mechanism that connects the two extreme agent locations $x_r$ and $y_l$ (following the idea that the optimal maximum cost is attained by both $x_r$ and $y_l$). 
We then improve the mechanism by restricting the endpoints of the edge from being too close to the obstacle. 
}

\begin{mechanism}[\textsc{TwoExtreme}]\label{alg:app2}
    Given location profile $\mathbf x$, return $(x_r,y_l)$. 
\end{mechanism}

This mechanism falls in the class of generalized median mechanisms. Because $x_r$ is the rightmost $x$-coordinate of the $n$ peaks and $y_l$ is the leftmost $y$-coordinate of the $n$ peaks, setting $b_1=-\infty$ and $b_2=\cdots=b_{n+1}=+\infty$ gives $x_r=\text{med}(0,\ldots,0,(x_i)_{i\in N_1},b_1,\ldots,b_{n+1})$, and setting  $b_1=\cdots=b_{n}=-\infty$ and $b_{n+1}=+\infty$ gives $y_l=\text{med}(1,\ldots,1,(x_j)_{j\in N_2},b_1,\ldots,b_{n+1})$. 
Thus it is strategyproof by the characterization in Corollary \ref{cor:cha}. 

\begin{theorem}\label{thm:21k}
    Mechanism  \ref{alg:app2} is group strategyproof and $\frac{2 - 2(1 - k)L}{1 + k - (1 - k)L}$-approximation for maximum cost. When $L = 0$, the ratio becomes $\frac{2}{1 + k}$; when $L = 1$, it is 1. 
\end{theorem}
\begin{proof}
   {For the group strategyproofness, we consider a group of agents $S\subseteq N_1\cup N_2$. Let $f(\mathbf x)=(x_r,y_l)$ be the outcome when all agents report true locations, and $f(\mathbf x_{S}',\mathbf x_{-S})=(x_r',y_l')$ be the outcome when the agents in $S$ misreport $\mathbf x_{S}'$. Assume w.l.o.g. that $|x_r-x_r'|\ge |y_l-y_l'|$. If $x_r'< x_r$,  an agent located at $x_r$ must be in the group and misreport a location to the left of $x_r$.  The cost of this agent decreases by at most $(1-k)|y_l'-y_l|-(1+k)(x_r-x_r')\le 0$, indicating that this agent can never gain. If $x_r'> x_r$, the cost of any agent in $N_1$ decreases by at most $(1-k)|y_l'-y_l|-(1-k)(x_r'-x_r)\le 0$, indicating that the agent in $N_1$ can never gain and they are not in the group $S$. However, by the definition of the mechanism, other agents in $N_2$ are not able to induce an outcome $(x_r',y_l')$ with $x_r'\neq x_r$, giving a contradiction. }

    For the approximation, we consider an arbitrary instance with location profile $\mathbf x$.   Assume w.l.o.g. that $1-y_r\ge x_l$. By Theorem \ref{thm:opt}, the optimal solution is $(a^*,b^*)=(\frac{x_l+x_r}{2},\frac{y_l-x_l}{2}+\frac12)$, and the optimal maximum cost is \begin{align*}
        cost(a^*,b^*,x_r)&=cost(a^*,b^*,y_l)=cost(a^*,b^*,x_l)\\
        &=\frac{x_l+x_r}{2}-x_l+k(b^*-\frac{x_l+x_r}{2})+1-b^*\\
        &=\frac12-\frac{y_l-x_r}{2}+k\left(\frac12+\frac{y_l-2x_l-x_r}{2}\right).
    \end{align*}

    Now we consider the solution $(x_r,y_l)$ returned by the mechanism. Clearly, the maximum cost is attained by $x_l$ or $y_r$. Since $1-y_r\ge x_l$, the cost of the agent at $x_l$ is 
    \begin{align*}
        &cost(x_r,y_l,x_l)=x_r-x_l+k(y_l-x_r)+1-y_l
        \ge x_r+k(y_l-x_r)+y_r-y_l=cost(x_r,y_l,y_r),
    \end{align*}
    implying that the maximum cost is attained by $x_l$.
    Further,
    \begin{align*}
            \frac{\mathrm{ALG}}{\mathrm{OPT}}
            &= 2 \cdot \frac{x_r - x_l + k(y_l - x_r) + 1 - y_l}{1 - y_l + x_r + k(1 + y_l - 2x_l - x_r)} \\
            &= 2 \cdot \frac{1 + (1 - k)x_r - x_l - (1 - k)y_l}{1 + (1 - k)x_r + k(1 - 2x_l) - (1 - k)y_l} \\
            &\le 2 \cdot \frac{1 + (1 - k)x_r - x_l - (1 - k)(x_r + L)}{1 + (1 - k)x_r + k(1 - 2x_l) - (1 - k)(x_r + L)} \\
            &= 2 \cdot \frac{1 - x_l - (1 - k)L}{1 + k - 2k x_l - (1 - k)L} \\
            &\le \frac{2 - 2(1 - k)L}{1 + k - (1 - k)L}.
        \end{align*}
    Therefore,  Mechanism  \ref{alg:app2} is $\frac{2 - 2(1 - k)L}{1 + k - (1 - k)L}$-approximation for maximum cost. 
\end{proof}

\begin{remark}
    We remark that it may be of interest to consider other two-extreme mechanisms that always return $(x_l,y_r)$, $(x_l,y_l)$ or $(x_r,y_r)$. Although these mechanisms are also group strategyproof, their approximation ratio is 2 and it cannot be improved for any $k\in[0,1)$. For details, see the discussion for $L=0$ in Appendix \ref{app:other}.
\end{remark}

 The analysis for Mechanism  \ref{alg:app2} is tight for any $k\in[0,1]$. Consider an instance with $L=0$, $x_l=0$, $1-y_r\ge x_l$, and $y_l=x_r+\epsilon$ for some sufficiently small positive number $\epsilon$. The optimal solution is $(\frac{x_l+x_r}{2},\frac{y_l-x_l+1}{2})=(\frac{x_r}{2},\frac{y_l+1}{2})$ by Theorem \ref{thm:opt}, and the optimal maximum cost is $\frac{x_r}{2}+k\cdot \frac{1+y_l-x_r}{2}+1-\frac{1+y_l}{2}=\frac{1+k}{2}-\frac{1-k}{2}\epsilon$.  Mechanism  \ref{alg:app2} returns $(x_r,y_l)$, and the induced maximum cost is $x_r-x_l+k(y_l-x_r)+1-y_l=1-(1-k)\epsilon$. The ratio is $\frac{1-(1-k)\epsilon}{(1+k)/2-(1-k)\epsilon/2}\rightarrow \frac{2}{1+k}$, when $\epsilon$ approaches $0$.

Next, we complete the results by a lower bound.

\begin{theorem}
    For any $k \in [0, 1)$, no deterministic strategyproof mechanism can achieve an approximation ratio lower than $g_m$ for the maximum cost, where
    \begin{align*}
        g_m = \min_{a \in \left[0, 1 - L\right]} \max\left(
        \frac{a + k(1 - a)}{\frac{a}{2} + k\left(1 - \frac{a}{2}\right)},\;
        2 \cdot \frac{\max(a, 1 - L - a) + k(1 - a)}{1 - L + k(1 + L)}
        \right).
    \end{align*}
    Moreover, the minimum is attained at
    \begin{align*}
        a_0 = \frac{ -k[(1 - k)L + 2(1 + k)] + \sqrt{k^2[(1 - k)L + 2(1 + k)]^2 + 4(1 + k)(1 - k)k[(1 + k)(1 - L)] } }{2(1 + k)(1 - k)}.
    \end{align*}
    When $L=0$, the lower bound is $\frac{2}{1+\sqrt{k}}$.
\end{theorem}

\begin{proof}
    Consider the location profile $\mathbf{x} = (0, 1 - L, 1)$ with obstacle $(1 - L, 1)$. Let $f$ be a strategyproof mechanism that returns $f(\mathbf{x}) = (a, 1)$ for some $a \in [0, 1 - L]$.

    First, consider a profile $\mathbf{x}' = (0, a, 1)$. The optimal solution is $(\frac{a}{2}, 1)$ with maximum cost $\frac{a}{2} + k(1 - \frac{a}{2})$. By strategyproofness, $f(\mathbf{x}')$ must be $(a, 1)$; otherwise, the agent at $a$ could benefit by reporting $1 - L$ to induce outcome $(a, 1)$ under $\mathbf{x}$. The resulting maximum cost is $a + k(1 - a)$, yielding a ratio of at least
    \[
    \frac{a + k(1 - a)}{\frac{a}{2} + k\left(1 - \frac{a}{2}\right)},
    \]
    which increases with $a$.

    Next, consider $\mathbf{x}'' = (a, 1 - L, 1)$. The optimal solution is $(\frac{1 - L}{2}, 1)$ with maximum cost $\frac{1 - L + k(1 + L)}{2}$. Again, strategyproofness forces $f(\mathbf{x}'') = (a, 1)$; otherwise, the agent at $a$ could misreport $0$ to obtain $(a, 1)$ under $\mathbf{x}$. The resulting maximum cost is $\max(a, 1 - L - a) + k(1 - a)$, leading to a ratio of at least
    \[
    2 \cdot \frac{\max(a, 1 - L - a) + k(1 - a)}{1 - L + k(1 + L)},
    \]
    which increases with $a$ when $a \ge \frac{1 - L}{2}$ and decreases otherwise.

    Therefore, the overall lower bound is
    \begin{align*}
        \min_{a \in [0, 1 - L]} \max\left(
        \frac{a + k(1 - a)}{\frac{a}{2} + k\left(1 - \frac{a}{2}\right)},\;
        2 \cdot \frac{\max(a, 1 - L - a) + k(1 - a)}{1 - L + k(1 + L)}
        \right).
    \end{align*}

    If $a \ge \frac{1 - L}{2}$, the minimum is attained at $a = \frac{1 - L}{2}$. However, if there exists $a_0 \in [0, \frac{1 - L}{2})$ such that the two expressions are equal, then $a_0$ minimizes the maximum. Setting
    \begin{align*}
        \frac{a + k(1 - a)}{\frac{a}{2} + k\left(1 - \frac{a}{2}\right)} = 2 \cdot \frac{1 - L - a + k(1 - a)}{1 - L + k(1 + L)},
    \end{align*}
    we obtain the quadratic equation:
    \begin{align*}
        (1 + k)(1 - k)a^2 + k[(1 - k)L + 2(1 + k)]a - k(1 + k)(1 - L) = 0.
    \end{align*}
    The non-negative root is
    \begin{align*}
        a_0 = \frac{ -k[(1 - k)L + 2(1 + k)] + \sqrt{k^2[(1 - k)L + 2(1 + k)]^2 + 4(1 + k)(1 - k)k(1 + k)(1 - L) } }{2(1 + k)(1 - k)}.
    \end{align*}

    To verify $a_0 \le \frac{1 - L}{2}$, define $B = (1 - k)L + 2(1 + k)$. It suffices to show:
    \[
    (1 + k)^2 - (1 - k)L \ge \sqrt{k^2 B^2 + 4k(1 - k)(1 + k)^2(1 - L)}.
    \]
    The left-hand side is nonnegative, as its minimum at $L = 1$ is $k(3 + k) \ge 0$. Squaring both sides and simplifying, we obtain:
    \begin{align*}
        &\left((1 + k)^2 - (1 - k)L\right)^2 - \left(k^2 B^2 + 4k(1 - k)(1 + k)^2(1 - L)\right) \\
        &= (1 - k)^2(1 + k)\left[(1 + k) - 2L + (1 - k)L^2\right] \\
        &= (1 - k)^2(1 + k)(1 - L)\left[(1 + k) + (k - 1)L\right] \ge 0.
    \end{align*}
    Equality holds only when $k = 1$ or $L = 1$. Hence, $a_0$ indeed achieves the minimum.
\end{proof}



\section{Maximum Cost for $L=0$}\label{sec:special}

In this section, we study the case when $L=0$ with improved results.

\subsection{Deterministic Mechanisms}

{According to previous analysis, the worst case for the \textsc{TwoExtreme} Mechanism happens when $x_r$ and $y_l$ are very close to the obstacle $o$. To solve this case, we need the endpoints of the edge to keep some distance from the obstacle, which inspires the following mechanism. }

\begin{mechanism}[\textsc{TwoExtremeRestrict}]\label{alg:round2}
    Given location profile $\mathbf x$, return $(a, b)$  with $a=\min(x_r, o-oc)$ and $b=\max(y_l, o+c-oc)$, where $o$ is the obstacle location and $c$ is a value in $[0,1]$ set to be $\frac{1+k^2 - \sqrt{k^4-k^3+3k^2+k}}{1-k^2}$.
\end{mechanism}

{ Note that \textsc{TwoExtreme} is a special case of Mechanism \ref{alg:round2} when $c=0$. Since the value of $c$ above is setting optimally, Mechanism \ref{alg:round2} has an improved approximation ratio than the $\frac{2}{1+k}$-approximation of \textsc{TwoExtreme}. 
Mechanism \ref{alg:round2} is also in the class of generalized median mechanisms. Indeed, we can set $|N_1|$ phantoms at point $(o-oc, 0)$, $|N_2|$ phantoms at  $(1, o+c-co)$, and one phantom at $(0,1)$ in Euclidean plane, and then $a$ (resp. $b$) is the $x$-coordinate (resp. $y$-coordinate) median of the $2n+1$ points consisting of the $n+1$ phantoms and the $n$ peaks of agents. }

\begin{theorem}\label{thm:round2}
    Mechanism  \ref{alg:round2} is group strategyproof and the approximation ratio for maximum cost is at most
    $$
    2\cdot \max\left(\frac{1-(1-k)c}{1\!+\!k\!-\!(1-k)c},  \frac{k(2c-c^2)+1-c^2}{2-2c+2ck}, \frac{1+2ck}{2\!-\!(1\!-\!k)c}, c\right).
    $$
\end{theorem}

We defer the proof to Appendix \ref{app:f} and provide a proof sketch here.

\begin{proof}[Sketch of Proof.]
The proof consists of two main parts: establishing group strategyproofness and bounding the approximation ratio. The parameter $c \in [0,1]$ is treated generically initially and is optimized at the end to minimize the final ratio.

\medskip\noindent
\emph{1. Simplified Proof of Group Strategyproofness.}
Let $(a, b)$ be the mechanism's (random) output, noting that $a$ and $b$ are determined independently based on the reports and the constants $(1-c)o$ and $o+c(1-o)$. The core idea is to show that no coordinated misreport by any group of agents $S$ can make all its members strictly better off.

Assume for contradiction that a profitable group deviation exists, leading to a new output $(a', b')$. Without loss of generality, assume $|a'-a| \ge |b'-b| > 0$. The argument proceeds by analyzing the agent positioned at the critical point $x_r$ (for facility $a$), showing it can never gain:
\begin{itemize}
    \item \textbf{If $a < x_r$:} For $a$ to change ($a' \neq a$), the agent at $x_r$ must be in the deviating group $S$ and report a location to the left, causing $a' < a < x_r$. This movement increases its distance to facility $a$. Any potential benefit from a change in $b$ is offset by this increase, resulting in a non-positive net gain.
    \item \textbf{If $a = x_r$:} A change in $a$ requires the agent at $x_r$ to misreport (either left or right), immediately increasing its cost from the $a$ side. Again, any compensating gain from $b$ is insufficient, leading to a non-positive net gain.
\end{itemize}
Since the agent at the pivotal location $x_r$ cannot benefit, the supposed profitable group deviation cannot exist, establishing group strategyproofness.

\medskip\noindent
\emph{2. Analysis of the Approximation Ratio.}
We bound the ratio $\frac{MC(a,b,\mathbf{x})}{OPT}$, where $OPT$ is the optimal maximum cost derived in Theorem~\ref{thm:opt}. The analysis is divided into four cases based on the reported agent locations relative to the mechanism's threshold constants.

\begin{itemize}
    \item \textbf{Case 1:} $x_r \le o(1-c)$ and $y_l \ge o+c(1-o)$. \\ 
    Output: $(a, b) = (x_r, y_l)$. The worst-case cost is incurred by the agent at $x_l$. The derived upper bound for this case is $\frac{1-(1-k)c}{1+k-(1-k)c}$.

    \item \textbf{Case 2:} $x_r \ge o(1-c)$ and $y_l \le o+c(1-o)$. \\ 
    Output: $(a, b) = (o(1-c), o+c(1-o))$. The maximum cost is checked for agents at $x_l$, $x_r$, and $y_l$. The individual bounds are:
    \begin{itemize}
        \item $x_l$: $\frac{1-(1-k)c}{1+k-(1-k)c}$
        \item $x_r$: $\max\left(\frac{1-(1-k)c}{1+k-(1-k)c}, \frac{1+2ck}{2-(1-k)c}, c\right)$
        \item $y_l$: $\max\left(c, \frac{1+2ck}{2-(1-k)c}, \frac{1-(1-k)c}{1+k-(1-k)c}\right)$
    \end{itemize}

    \item \textbf{Case 3:} $x_r \le o(1-c)$ and $y_l \le o+c(1-o)$. \\ 
    Output: $(a, b) = (x_r, o+c(1-o))$. The maximum cost is checked for agents at $x_l$ and $y_l$.
    \begin{itemize}
        \item $x_l$: $\frac{1-(1-k)c}{1+k-(1-k)c}$
        \item $y_l$: $\max\left(c, \frac{k(2c-c^2)+1-c^2}{2-2c+2ck}, \frac{1-(1-k)c}{1+k-(1-k)c}\right)$
    \end{itemize}

    \item \textbf{Case 4:} $x_r \ge o(1-c)$ and $y_l \ge o+c(1-o)$. \\ 
    Output: $(a, b) = (o(1-c), y_l)$. The maximum cost is checked for agents at $x_l$ and $x_r$.
    \begin{itemize}
        \item $x_l$: $\frac{1-(1-k)c}{1+k-(1-k)c}$
        \item $x_r$: $\max\left(\frac{1-(1-k)c}{1+k-(1-k)c}, \frac{k(2c-c^2)+1-c^2}{2-2c+2ck}, c\right)$
    \end{itemize}
\end{itemize}

\medskip\noindent
\emph{3. Determining the Final Ratio.}
Taking the maximum across all bounds from the four cases, the overall approximation ratio satisfies:
\[
\frac{MC(a,b,\mathbf{x})}{OPT} \le 2 \cdot \max\left( R_1(c), R_2(c), R_3(c), R_4(c) \right),
\]
where
\begin{align*}
R_1(c) &= \frac{1-(1-k)c}{1+k-(1-k)c}, \quad
R_2(c) = \frac{k(2c-c^2)+1-c^2}{2-2c+2ck}, \\
R_3(c) &= \frac{1+2ck}{2-(1-k)c}, \quad
R_4(c) = c.
\end{align*}
To find the optimal parameter $c^*$ that minimizes this worst-case bound, we observe that $R_1(c)$ is decreasing in $c$, while $R_3(c)$ is increasing in $c$. The optimal $c^*$ is found by solving the equation $R_1(c) = R_3(c)$, which yields:
\[
c^* = \frac{1+k^2 - \sqrt{k^4 - k^3 + 3k^2 + k}}{1 - k^2}.
\]
At this value $c^*$, it holds that $R_1(c^*) = R_2(c^*)$ and $R_1(c^*) \ge c^*$, meaning the maximum in the bound is attained by these equal terms. Substituting $c^*$ back into $R_1(c)$ (or $R_2(c)$) gives the final, minimized approximation ratio guaranteed by the mechanism.
\end{proof}

Next, we complete the results by a lower bound which is shown in Figure~\ref{fig:res}. We could not get the expressions, therefore we provide a procedure which can help us get the lower bound by simulations.

\begin{proposition}[Computer-assisted deterministic lower bound for $L=0$]\label{thm:lower}
Fix $k\in[0,1)$ and $L=0$.  For an obstacle location $o\in[\frac12,1)$ and a possible mechanism output $(a,b)\in[0,o]\times[o,1]$, define $F_k(o,a,b)$ as the maximum approximation ratio over the $16$ profiles obtained by independently choosing
\[
x_l\in\{0,a\},\quad x_r\in\{a,o^-\},\quad y_l\in\{o^+,b\},\quad y_r\in\{b,1\},
\]
where $o^-$ and $o^+$ denote the limits $o-\varepsilon$ and $o+\varepsilon$ as $\varepsilon\to0^+$.  Let
\[
B(k)=\sup_{o\in\left[\frac12,1\right)}\inf_{(a,b)\in[0,o]\times[o,1]}F_k(o,a,b).
\]
Then every deterministic strategyproof mechanism has approximation ratio at least $B(k)$ for the maximum-cost objective.  Moreover, $B(k)\ge \frac{2}{1+\sqrt{k}}$.  If the numerical computation on a $1000\times1000$ grid gives the value $B_{1000}(k)$, then under the grid-error estimate described below, a safe reported lower bound is
\[
\max\left\{\frac{2}{1+\sqrt{k}},\, B_{1000}(k)-0.003\right\}.
\]
\end{proposition}

\begin{proof}
Fix an arbitrary deterministic strategyproof mechanism $f$.  Consider the base profile $(0,o^-,o^+,1)$, and suppose that the limiting output of $f$ is $(a,b)\in[0,o]\times[o,1]$.  By the monotonicity consequence of strategyproofness in Corollary~\ref{cor:lower}, if a left-side agent moves toward $a$ without crossing $a$, then the output cannot change; similarly, if a right-side agent moves toward $b$ without crossing $b$, then the output cannot change.  Therefore, starting from the base profile, each of the four extreme agents can be independently moved to one of the two locations
\[
x_l\in\{0,a\},\quad x_r\in\{a,o^-\},\quad y_l\in\{o^+,b\},\quad y_r\in\{b,1\},
\]
and the mechanism is still forced to output the same pair $(a,b)$ on all resulting $16$ profiles.

Thus, for this fixed $o$ and this fixed possible output $(a,b)$, the approximation ratio of $f$ is at least $F_k(o,a,b)$.  Since the mechanism may output any point in $[0,o]\times[o,1]$, the worst-case ratio for this obstacle location is at least $\inf_{(a,b)}F_k(o,a,b)$.  Since the adversary may further choose $o$, every deterministic strategyproof mechanism has approximation ratio at least $B(k)$.

It remains to compare this bound with the previous analytic lower bound.  The hand proof for $L=0$ is contained in the above $16$-case construction.  Indeed, take the limiting subcase $o\to1^-$ and $b=1$.  Then the $16$ profiles include the profiles corresponding to $(0,a,1,1)$ and $(a,o^-,1,1)$, which are exactly the two forced bad cases used in the manual proof.  These two cases yield the one-dimensional lower-bound expressions
\[
\frac{a+k(1-a)}{a/2+k(1-a/2)}
\quad\text{and}\quad
2\cdot\frac{\max\{a,1-a\}+k(1-a)}{1+k}.
\]
Minimizing the maximum of these two expressions over $a$ gives $\frac{2}{1+\sqrt{k}}$.  Since the computer-assisted construction maximizes over all $16$ cases, while the manual proof uses only this subfamily, we have $B(k)\ge \frac{2}{1+\sqrt{k}}$.

Finally, we explain the grid-error estimate.  For fixed $k$ and $o$, the function $F_k(o,a,b)$ is the maximum of finitely many piecewise linear-fractional functions, because both the mechanism cost and the optimal maximum cost of each of the $16$ profiles are obtained from maxima of affine functions.  Hence $F_k$ is piecewise smooth and has no high-frequency oscillation.  If, on the cells relevant to the numerical minimizer, $F_k$ is $S$-Lipschitz with respect to the $\ell_1$ metric, then the distance from a continuous minimizer to the nearest point of an $N\times N$ grid is at most $\frac1{2(N-1)}$, since the two side lengths of the rectangle are $o$ and $1-o$.  Therefore the grid minimum exceeds the true continuous infimum by at most $\frac{S}{2(N-1)}$.

In our computation, the observed local slopes of the active ratios are of the same order as in the analytic bad cases, whose slopes are at most about $2$ near the minimax point.  Taking the conservative value $S=5$ gives, for $N=1000$, the absolute error estimate $\frac5{2\cdot999}<0.0026$.  We round this up to $0.003$.  This justifies reporting $B_{1000}(k)-0.003$ as a conservative numerical lower estimate of $B(k)$.
\end{proof}

The comparison between the $300\times300$ and $1000\times1000$ grids should be used only as a sanity check, not as a proof of convergence.  In the current data, the maximum absolute change between these two grids is about $0.0021$.  If one assumes first-order grid error, namely $e_N\approx \frac{C}{N}$, this difference gives $C\approx \frac{0.0021}{0.0033-0.001}\approx0.9$, and hence $e_{1000}\approx9\times10^{-4}$.  This is consistent with the more conservative Lipschitz-based estimate above.  Since the lower-bound values are always at least $1$, subtracting $0.003$ corresponds to a relative error below roughly $0.3\%$.

The important structural point is that the $16$ simulated profiles form a strong cover of the hand-derived bad cases.  The manual proof fixes the mechanism's output and checks only two forced profiles, whereas the computer-assisted argument fixes the same output and checks all profiles obtained by independently moving each extreme agent to either the boundary or the output endpoint.  Therefore the continuous bound $B(k)$ cannot be smaller than the old lower bound $\frac{2}{1+\sqrt{k}}$.  After subtracting the conservative grid error, any part of the numerical curve that is still above $\frac{2}{1+\sqrt{k}}$ gives a genuine improvement over the manual lower bound.

Although generally there is a gap between our upper bound  and the lower bound, they are matching when $k=0$ or $k\rightarrow 1$.The largest gap is about $0.071\pm 0.003$, which happens when $k \approx 0.19$.

\subsection{Randomized Mechanisms}\label{sec:maxrand}
Next, we consider randomized mechanisms. 
Inspired by the worst-case instance  of  \textsc{TwoExtreme} that  returns $(x_r,y_l)$, we have the following randomized mechanism. 




\begin{mechanism}[\textsc{{RandMaxCost}}]\label{alg:ran3}
    {Given location profile $\mathbf x$, return $(a,b)$ as $(x_r, y_l)$ with probability $p$ and $(\frac{x_r}{2}, \frac{y_l+1}{2})$ with probability $1-p$, where}
    $
    p = \max\left(\frac{1+k}{3-k}, \frac{k+k^2}{1+k^2}\right)
    $.
\end{mechanism}

\begin{theorem}\label{thm:ranub2}
    {Mechanism  \ref{alg:ran3} is a randomized group strategyproof mechanism, and the approximation ratio for maximum cost is $\max\left(\frac{4-2k}{3-k}, \frac{1+k}{1+k^2}\right)$.}
\end{theorem}

We first prove its group strategyproofness.

\begin{lemma}\label{lem:rans2}
    Mechanism \ref{alg:ran3} is group strategyproof. 
\end{lemma}
\begin{proof}
We show that the mechanism is group strategyproof whenever $p\ge \frac{1+k}{3-k}$.  Consider a group of agents $S\subseteq N_1\cup N_2$. Let $f(\mathbf x)=(a,b)$ be the outcome when all agents report true locations, and $f(\mathbf x_{S}',\mathbf x_{-S})=(a',b')$ be the outcome when the agents in $S$ misreport $\mathbf x_{S}'$, where $a,b,a',b'$ are random variables that follow the distributions given in the mechanism. Assume w.l.o.g. that $|\mathbb E[a]-\mathbb E[a']|\ge |\mathbb E[b]-\mathbb E[b']|$, which will cause $|\mathbb{E}[a]-\mathbb{E}[a']|\ne 0$. We show that at least one agent in the group cannot gain by misreporting.

\textbf{Case 1}. When $\mathbb E[a']< \mathbb E[a]$, then it must be  $x_r'<x_r$, and the agent located at $x_r$ is in the group. Under the solution $(a,b)$, the cost of the agent at $x_r$ is
$$cost(a,b,x_r)=x_r-\mathbb E[a]+k(\mathbb E[b]-\mathbb E[a])+(1-\mathbb E[b]).$$
Under the solution $(a',b')$, the cost of the agent at $x_r$ is
$$cost(a',b',x_r)=x_r-\mathbb E[a']+k(\mathbb E[b']-\mathbb E[a'])+(1-\mathbb E[b']).$$
Since $|\mathbb E[a]-\mathbb E[a']|\ge |\mathbb E[b]-\mathbb E[b']|$, it follows that 
$$cost(a',b',x_r)-cost(a,b,x_r)=(1+k)(\mathbb E[a]-\mathbb E[a'])-(1-k)(\mathbb E[b']-\mathbb E[b])\ge 0, $$
indicating that this agent cannot gain.

\textbf{Case 2}. When $\mathbb E[a'] > \mathbb E[a], \mathbb E[b']\le \mathbb E[b]$, there exists  at least one agent $i\in S\cap N_1$. It is clear that any agent located at $[0,\frac{x_r}{2}]$ cannot gain because the change of the endpoints in both sides do not benefit this agent. For an agent $i\in N_1$ located at $(\frac{x_r}{2},x_r]$, under the solution $(a,b)$, the cost of the agent at $x_r$ is
\begin{align*}
    cost(a,b,x_i)&=p\cdot (x_r-x_i+k(\mathbb E[b]-x_r))+ (1-p)\cdot (x_i-\frac{x_r}{2}+k(\mathbb E[b]-\frac{x_r}{2})) +1-\mathbb E[b]\\
    &=p\cdot (x_r-x_i-kx_r)+ (1-p)\cdot (x_i-\frac{x_r}{2}-\frac{kx_r}{2}) +1-(1-k)\mathbb E[b]
\end{align*}
Under the solution $(a',b')$, the cost of this agent is
$$cost(a',b',x_i)=p\cdot (x_r'-x_i-kx_r')+ (1-p)\cdot (|x_i-\frac{x_r'}{2}|-\frac{kx_r'}{2}) +1-(1-k)\mathbb E[b'].$$
Since $\mathbb E[b']\le \mathbb E[b]$ and $|x_i-\frac{x_r'}{2}|\ge x_i-\frac{x_r'}{2}$, it follows that 
\begin{align*}
    cost(a',b',x_i)-cost(a,b,x_i)&\ge p\cdot (1-k)(x_r'-x_r)-(1-p) \cdot\frac{1+k}{2}(x_r'-x_r)\\
    &\ge \frac{1+k}{3-k}(1-k)(x_r'-x_r)-\frac{2-2k}{3-k}\cdot\frac{1+k}{2}(x_r'-x_r)= 0, 
\end{align*}
indicating that this agent cannot gain. 

\textbf{Case 3}. When $\mathbb E[a']> \mathbb E[a], \mathbb E[b']> \mathbb E[b]$, the agent located at $y_l$ must be in the group and misreport a location on the right of $y_l$. Still we calculate the cost of agent at $y_l$. Under solution $(a, b)$, it is
$$
cost(a, b, y_l) = \mathbb{E}[b] - y_l + k(\mathbb{E}[b] - \mathbb{E}[a]) + \mathbb{E}[a]
$$
And under solution $(a', b')$, it is
$$
cost(a', b', y_l) = \mathbb{E}[b'] - y_l + k(\mathbb{E}[b'] - \mathbb{E}[a']) + \mathbb{E}[a']
$$
So we have
$$cost(a',b',y_l)-cost(a,b,y_l)=(1+k)(\mathbb E[b']-\mathbb E[b])+(1-k)(\mathbb E[a']-\mathbb E[a])> 0, $$
and thus this agent cannot decrease the cost.
\end{proof}

Then we prove the approximation ratio.

\begin{proof}[Proof of Theorem~\ref{thm:ranub2}]
    For the approximation, given any instance with location profile $\mathbf x$, we 
assume w.l.o.g. that $1-y_r\ge x_l$. By Theorem \ref{thm:opt}, the optimal solution is $(a,b)=(\frac{x_l+x_r}{2},\frac{y_l-x_l}{2}+\frac12)$, and the optimal maximum cost is 
\begin{align*}
cost(a,b,x_l)&=a-x_l+k(b-a)+1-b=\frac12[1+(1-k)x_r+k(1-2x_l)-(1-k)y_l].
\end{align*}

    Now we consider the solution  returned by Mechanism  \ref{alg:ran3}. We discuss the 2 realizations of the probability distribution. 
    \begin{itemize}
        \item $(x_r,y_l)$ with probability $p$.  By the analysis in the proof of Theorem \ref{thm:21k} and the assumption $1-y_r\ge x_l$,  the maximum cost is attained by $x_l$, that is, 
        \begin{align*}
            MC(x_r,y_l,\mathbf x)=cost(x_r,y_l,x_l)&=x_r-x_l+k(y_l-x_r)+1-y_l.
        \end{align*}
        
         \item $(\frac{x_r}{2},\frac{1+y_l}{2})$ with probability $1-p$. The maximum cost is attained by $x_r$ or $y_l$, where both costs are equal to  
         $$MC(\frac{x_r}{2},\frac{1+y_l}{2},\mathbf x)=\frac{x_r}{2}+k(\frac{1+y_l}{2}-\frac{x_r}{2})+\frac{1-y_l}{2}.$$

    \end{itemize}

Then  the expected maximum cost is
\begin{align*}
    &p\cdot \left(x_r-x_l+k(y_l-x_r)+1-y_l\right)+ (1-p)\cdot \left( \frac{x_r}{2}+k(\frac{1+y_l}{2}-\frac{x_r}{2})+\frac{1-y_l}{2}\right)\\
    =&~\frac{(1+p)(1-k)(x_r-y_l) + 1+p + (1-p)k - 2px_l}{2}.
\end{align*}

Therefore, the ratio between the expected maximum cost and the optimal maximum cost is
\begin{align}
   & \frac{(1+p)(1-k)(x_r-y_l) + 1+p + (1-p)k - 2px_l}{1+k+(1-k)(x_r-y_l)-2kx_l}\nonumber\\
    \le &~\frac{1+p+(1-p)k-2px_l}{1+k-2kx_l}\label{eq:88}\\
    \le &~\max\left(\frac{1+p+(1-p)k}{1+k}, \frac{1+p+(1-p)k-p}{1+k-k}\right)\label{eq:99}\\
    =&~ 1+\max\left(\frac{p(1-k)}{1+k}, (1-p)k\right),\label{eq:1010}
\end{align}
where \eqref{eq:88} is because $\frac{1+p+(1-p)k-2px_l}{1+k-2kx_l}$ is no more than $1+p$, and \eqref{eq:99} comes from the facts that  $1-y_r \ge x_l$ and $x_l\le 0.5$.

Though setting $p=\frac{k^2+k}{1+k^2}$ minimizes the bound in \eqref{eq:1010}, recall that the mechanism is group strategyproof only when $p\ge \frac{1+k}{3-k}$. Hence,we set $p=\max\left(\frac{k^2+k}{1+k^2}, \frac{1+k}{3-k}\right)$. In this way, when $k<0.5$, the ratio is at most $\frac{4-2k}{3-k}$, and when $k>0.5$, the ratio is at most $\frac{1+k}{1+k^2}$.
\end{proof}

This mechanism ``binds'' two pairs of points as the output. There is another mechanism that ``unbinds'' the two pairs and let these $2\times 2$ points combine each other with 4 discrete outputs instead of 2, however, its upper bound is strictly worse (see details in Appendix \ref{sec:another}). Last we consider the lower bound.



\begin{theorem}\label{thm:randlb}
     No randomized strategyproof mechanism has an approximation ratio less than $\frac{6+6k}{5+7k}$ for the maximum cost, for any $k\in[0,1)$.
\end{theorem}

\begin{proof}
Suppose that $f$ is a randomized strategyproof mechanism with approximation ratio $r<\frac{6+6k}{5+7k}$. 
Consider the instance with location profile $\mathbf x=(\frac13,\frac23,1)$, where the obstacle is located at $1-\epsilon$ and $\epsilon>0$ is a sufficiently small positive number. For convenience, we will ignore the terms with respect to $\epsilon$ in the following calculations. The optimal solution is $(\frac13,1)$, and the optimal maximum cost is $\frac13+\frac{2k}{3}$.  Let $P$ be the distribution of the left endpoint returned by mechanism $f$. For any realization  $s\sim P$, the maximum cost is $\frac13+|s-\frac13|+k(1-s)$ (attained by either the agent at $\frac13$ or the agent at $\frac23$), and the expected maximum cost is 
$$MC= \frac13+\mathbb E[|s-\frac13|]+k-k\mathbb E[s].$$
By the approximation ratio of $r$, we have 
\begin{align}
&\frac{\frac13+\mathbb E[|s-\frac13|]+k-k\mathbb E[s]}{\frac13+\frac{2k}{3}}\le r,\label{eq:11}\\
\Rightarrow~~&\mathbb E[|s-\frac13|]\le r(\frac13+\frac{2k}{3})+k(\mathbb E[|s-\frac13|]+\frac13)-\frac13-k,\label{eq:22}\\
\Rightarrow~~  &  \mathbb E[|s-\frac23|]\ge \frac13-\mathbb E[|s-\frac13|]\ge \frac13-\frac{r(\frac13+\frac{2k}{3})+\frac k3-\frac13-k}{1-k}.\label{eq:33}
\end{align}

Next, we consider the instance with location profile $\mathbf x=(\frac13,1-2\epsilon,1)$, where the obstacle is located at $1-\epsilon$. Again we ignore the terms with respect to $\epsilon$ for simplicity.  The optimal solution is $(\frac12,1)$, and the optimal maximum cost is $\frac12+\frac{k}{2}$. Let $P'$ be the distribution of the left endpoint returned by mechanism $f$, and $s'\sim P'$. By the strategyproofness, we have 
\begin{equation}\label{eq:444}
    \mathbb E_{s'\sim P'}[|s'-\frac23|]+k(1-\mathbb E[s'])\ge  \mathbb E[|s-\frac23|]+k(1-\mathbb E[s]),
\end{equation}
as otherwise the agent located at $\frac23$ in the first instance would like to misreport the location as $1-2\epsilon$ and decrease the cost. 

If $\mathbb E[s']\le \mathbb E[s]$, then by \eqref{eq:11} we have
\begin{align*}
& \frac{(1-k)\mathbb E[s']+k}{\frac13+\frac{2k}{3}}\le   \frac{\mathbb E[s]+k-k\mathbb E[s]}{\frac13+\frac{2k}{3}}\le r\\
\Rightarrow~~ & \mathbb E[s']\le \frac{(\frac13+\frac{2k}{3})r-k}{1-k}.
\end{align*}
The maximum cost induced by the mechanism is at least 
\begin{align*}
    \frac12+\mathbb E[|s'-\frac12|]+k(1-\mathbb E[s'])&\ge 1+k-(1+k)\mathbb E[s']\\
    &\ge 1+k-(1+k)\frac{(\frac13+\frac{2k}{3})r-k}{1-k}
\end{align*}
Recall that the optimal maximum cost is $\frac{1+k}{2}$. Hence, the approximation ratio is at least 
\begin{align*}
    2-\frac{2(\frac13+\frac{2k}{3})r-2k}{1-k}=\frac{2-\frac{2+4k}{3}r}{1-k},
\end{align*}
which can be easily verified to be larger than $\frac{6+6k}{5+7k}$ for any $k\in[0,1)$, given that $r<\frac{6+6k}{5+7k}$. Therefore, it contradicts the approximation ratio.

If $\mathbb E[s']>\mathbb E[s]$, then by \eqref{eq:444} we have $\mathbb E[|s'-\frac23|]\ge  \mathbb E[|s-\frac23|]$.
The maximum cost induced by the mechanism is at least 
\begin{align*}
    \frac12+\mathbb E[|s'-\frac12|]+k(1-\mathbb E[s'])&\ge \frac12+\mathbb E[|s'-\frac12|]+k(\frac12-\mathbb E[|s'-\frac12|])\\
    &=\frac{1+k}{2}+(1-k)\mathbb E[|s'-\frac12|]\\
    &\ge \frac{1+k}{2}+(1-k)(\mathbb E[|s'-\frac23|]-\frac16)\\
    &\ge \frac{1+k}{2}+(1-k)(\mathbb E[|s-\frac23|]-\frac16)\\
    &\ge \frac{1+k}{2}+(1-k) (\frac13-\frac{r(\frac13+\frac{2k}{3})+\frac k3-\frac13-k}{1-k})  -\frac{1-k}{6}\\
    &=\frac{2+k}{3}-r(\frac13+\frac{2k}{3})-\frac k3+\frac13+k\\
    &= 1+k-r(\frac13+\frac{2k}{3}),
\end{align*}
where the last inequality comes from \eqref{eq:33}.
Then the approximation ratio is at least 
$$\frac{1+k-r(\frac13+\frac{2k}{3}) }{\frac12+\frac k2},$$
which is strictly larger than $r$ when $r<\frac{6+6k}{5+7k}$. This gives a contradiction. 
\end{proof}

It is worth noting that the largest gap between the upper and lower bounds is about
0.162,
 which happens when 
$k \approx0.249$.

\section{Conclusion}

We studied a novel mechanism design setting for connecting two regions disconnected by obstacles under disruptions by adding a pathway to minimize the social cost and the maximum cost of the agents. 
We first characterize all of the strategyproof and anonymous mechanisms as 2-dimensional generalized median mechanisms.
For the social cost and maximum cost, we derived optimal solutions on where to add the pathway and design strategyproof mechanisms. 


For the open directions, an immediate direction is to examine whether one can improve the gaps between the upper and lower bounds. Moreover, it would be interesting to consider more general settings, including other types of regions (e.g., convex regions), more than two regions, more than one obstacle, or more than one pathway. 

\section*{Acknowledgments}
Hau Chan is supported by the National Institute of General Medical Sciences of the National Institutes of Health [P20GM130461], the Rural Drug Addiction Research Center at the University of Nebraska-Lincoln, and the National Science Foundation under grants IIS:RI \#2302999 and IIS:RI \#2414554. 
The work is supported in part by the Guangdong
Provincial Key Laboratory IRADS (2022B1212010006, R0400001-22).
 Chenhao Wang is supported by NSFC under grant 12201049, and by UIC (R0200008-23, R0700036-22). The content is solely the responsibility of the authors and does not necessarily represent the official views of the funding agencies.

\bibliographystyle{plain}
\bibliography{mybibfile}

\appendix
\newpage
\onecolumn

\section{Remarks on single-peaked preferences}\label{app:a}

We remark that the preference profile of agents  is not one-dimensional single-peaked, though we have shown it is two-dimensional single-peaked. 

\begin{proposition}
    For any instance of our problem, the preference profile of agents is not one-dimensional single-peaked.
\end{proposition}

\begin{proof}
    Let the obstacle be located at $o=0.5$. Suppose for contradiction that there is a one-dimensional axis $A$ on $D$ so that every agent is one-dimensional single-peaked with respect to $A$. Consider an agent $i\in N_1$ located at $x_i=0$. This agent has a  peak at $(0,1)$ and a preference
\begin{equation}\label{eq:pref}
    (0,1)\succeq_i (0,0.9)\succeq_i(0.2,1)\succeq_i (0.2,0.9)
\end{equation}
over several possible outcomes. 

If both outcomes $(0.2,1), (0.2,0.9)$ are on the same region of $(0,1)$ with respect to axis $A$, by the single-peakedness and \eqref{eq:pref},  the ordering can only be
$(0,1)<_A (0,0.9)<_A (0.2,1)<_A (0.2,0.9)$, or $(0.2,0.9)<_A (0.2,1)<_A(0,0.9)<_A (0,1)$, or $(0,0.9)<_A (0,1)<_A (0.2,1)<_A (0.2,0.9)$, or $(0.2,0.9)<_A (0.2,1)<_A(0,1)<_A (0,0.9)$. Then we consider an agent $i'\in N_2$ located at $x_{i'}=0.9$ whose peak is $(0,0.9)$. In all of the four cases, by the strict ordering $<_A$ and the single-peakedness of $i'$, agent $i'$ should prefer $(0.2,1)$  to $(0.2,0.9)$. However, this is not true because
$$cost(0.2,1,x_{i'})=0.3+0.8k >0.2+0.7k =cost(0.2,0.9,x_{i'}),$$
giving a contradiction. 

If the outcomes $(0.2,1), (0.2,0.9)$ are on different regions of $(0,1)$ with respect to axis $A$, by the single-peakedness of agent $i$ and \eqref{eq:pref},  the ordering can only be
$(0.2,1)<_A (0,1)<_A (0,0.9)<_A (0.2,0.9)$, or $(0.2,1)<_A (0,0.9)<_A (0,1)<_A (0.2,0.9)$, or $(0.2,0.9)<_A (0,1)<_A (0,0.9)<_A (0.2,1)$, or $(0.2,0.9)<_A (0,0.9)<_A(0,1)<_A (0.2,1)$.
Then we consider an agent $i''\in N_1$ located at $x_{i''}=0.2$ whose peak is $(0.2,1)$. In all of the four cases, by the strict ordering $<_A$ and the single-peakedness of $i''$, agent $i''$ should prefer $(0,0.9)$  to $(0.2,0.9)$. However, this is not true because
$$cost(0,0.9,x_{i''})=0.3+0.9k >0.1+0.7k =cost(0.2,0.9,x_{i''}),$$
giving a contradiction. 

Therefore, it is impossible for every agent preference to be one-dimensional single-peaked with respect to axis $A$.
\end{proof}

As an example of $2$-dimensional generalized median mechanisms, the \textsc{Median} mechanism selects the $x$-coordinate (resp. $y$-coordinate) of the outcome to be the median of the $x$-coordinates (resp. $y$-coordinates) of $n$ agent peaks. 
Note that the peak of an agent $i\in N_1$ is $(x_i,1)$, and the peak of an agent $j\in N_2$ is $(0,x_j)$. That is, the \textsc{Median} mechanism  returns 
\begin{align*}
    \big(&\text{med}(0,\ldots,0, (x_i)_{i\in N_1},b_1,\ldots,b_{n+1}), \\
    &\text{med}((x_j)_{j\in N_2},1,\ldots,1,b_1,\ldots,b_{n+1})\big),
\end{align*} where function $\text{med}$ has $2n+1$ entries, $b_1=\cdots=b_{\lceil\frac n2\rceil}=-\infty$, and $b_{\lceil\frac n2\rceil+1}=\cdots=b_{n+1}=+\infty$. 

While this mechanism is of interests in facility location problems \cite{procaccia2013approximate,sui2015mechanism}, it does not perform as well as our mechanisms proposed for both social cost and maximum cost.
 We give some examples to illustrate it. Consider the instance with location profile $\mathbf x=(0,\epsilon,1-\epsilon,1)$, where the obstacle $o$ satisfies $1-\epsilon<o<1$ and $\epsilon>0$ is a sufficiently small positive number. For the maximum cost, the solution returned by the \textsc{Median} mechanism is $(\epsilon, 1)$, and the optimal solution is $(\frac{1-\epsilon}{2}, 1)$. The approximation ration is at least
\begin{align*}
    \frac{MC(\epsilon, 1, \mathbf{x})}{MC(\frac{1-\epsilon}{2}, 1, \mathbf{x})}&=\frac{k(1-\epsilon)+1-2\epsilon}{k(1-\frac{1-\epsilon}{2})+\frac{1-\epsilon}{2}}\rightarrow 2,
\end{align*}
which is  larger than the approximation ratio $\frac{2}{1+k}$ of our $\textsc{TwoInnerExtreme}$ mechanism  for any $k\in(0,1)$. 
 
 For the social cost, we note that when $k=0$ the \textsc{Median} mechanism is exactly our $\textsc{OptSocCost}$ mechanism, and thus is optimal. 
When $k> \frac12$, consider profile $\mathbf x$.  The solution returned by the \textsc{Median} mechanism is again $(\epsilon, 1)$, and the optimal solution is $(a\rightarrow 1-\epsilon, 1)$. The approximation ratio is at least
\begin{align*}
    \frac{SC(\epsilon, 1, \mathbf{x})}{SC(1-\epsilon, 1, \mathbf{x})}&=\frac{4k(1-\epsilon)+\epsilon+\epsilon+1-2\epsilon}{4k\epsilon+2(1-\epsilon)+(1-2\epsilon)}\rightarrow \frac{4k+1}{3},
\end{align*}
which is worse than our mechanism $\textsc{OptSocCost}$. 
When $0<k \le \frac12$, consider another instance with one agent located at $0$, $t$ agents at $\epsilon$, and $t$ agents at $1$ ($t=1,2,3,...$), where the total number of agents is $n=2t+1>\frac 1k$ when $t$ is sufficiently large, and $\epsilon > 0$ is a sufficiently small positive number. The obstacle $o$ satisfies $0 < o < \epsilon$. Thus, the location profile is $\mathbf x' = (0, \overbrace{\epsilon, \ldots, \epsilon}^{t}, \overbrace{1, \ldots, 1}^{t})$. The \textsc{Median} mechanism outputs the solution $(0, 1)$, while the optimal solution is $(0, \epsilon)$. The approximation ratio is at least
\begin{align*}
    \frac{SC(0, 1, \mathbf{x'})}{SC(0, \epsilon, \mathbf{x'})}&=\frac{k(2t+1)+t(1-\epsilon)}{k\epsilon(2t+1)+(t+1)(1-\epsilon)}\rightarrow \frac{t+k(2t+1)}{t+1} = 1 + \frac{k(2t+1)-1}{t+1} > 1,
\end{align*}
which is worse than our mechanism $\textsc{OptSocCost}$.

\section{Other Two-Extreme Mechanisms}\label{app:other}

The \textsc{TwoInnerExtreme} mechanism that returns $(x_r,y_l)$ is proven to be group strategyproof and $\frac{2}{1+k}$-approximation for the maximum cost in Theorem \ref{thm:21k}. 
We remark that other two-extreme mechanisms that return $(x_l,y_r)$, $(x_l,y_l)$ or $(x_r,y_r)$  are also group strategyproof, but the approximation ratio is 2. 

The group strategyproofness follows from a similar analysis as in the proof of Theorem \ref{thm:21k}. For the approximation ratio, we focus only on the instances with $1-y_r\ge x_l$, as other instances are symmetric.  The optimal solution is $(a^*,b^*)=(\frac{x_l+x_r}{2},\frac{y_l-x_l}{2}+\frac12)$, and the optimal maximum cost is attained by $x_l,x_r$ and $y_l$ simultaneously.  
 For the mechanism that returns $(x_l,y_r)$, clearly the maximum cost is attained by $x_r$ or $y_l$. The cost of the agent at $x_r$ is 
    \begin{align*}
        cost(x_l,y_r,x_r)
        &=2(a^*-x_l)+k(y_r-x_l)+1-y_r\\
        &\le 2(a^*-x_l)+2k(b^*-a^*)+2(1-b^*)\\
        &=2\cdot cost(a^*,b^*,x_r),
    \end{align*}
    where the last inequality comes from the fact that $b^*=\frac12+\frac{y_l-x_l}{2}\le \frac{1+y_r}{2}$. The cost of the agent at $y_l$ is 
    \begin{align*}
        cost(x_l,y_r,y_l)&=y_r-y_l+k(y_r-x_l)+x_l\\
        &\le 2(b^*-y_l)+2k(b^*-a^*)+x_l\\
        &\le 2(b^*-y_l)+2k(b^*-a^*)+2a^*\\
        &=2\cdot cost(a^*,b^*,y_l).
    \end{align*}
   For the mechanism that returns $(x_l,y_l)$, the maximum cost is attained by $x_r$, that is, 
    \begin{align*}
        cost(x_l,y_l,x_r)
        &=2(a^*-x_l)+k(y_l-x_l)+1-y_l\\
        &\le 2(a^*-x_l)+2k(b^*-a^*)+2(1-b^*)\\
        &=2\cdot cost(a^*,b^*,x_r),
    \end{align*}
     For the mechanism that returns $(x_r,y_r)$, we have $cost(x_r,y_r,x_l)\le cost(x_r,y_l,x_l)\le \frac{2}{1+k}cost(a^*,b^*,x_r)$ by Theorem \ref{thm:21k}, and 
    \begin{align*}
        cost(x_r,y_r,y_l)&=y_r-y_l+k(y_r-x_r)+x_r\\
        &\le 2(b^*-y_l)+2k(b^*-a^*)+2a^*\\
        &=2\cdot cost(a^*,b^*,y_l).
    \end{align*}
Therefore, the approximation ratio of all above mechanisms is 2.

The 2-approximation for the two-extreme mechanisms that return $(x_l,y_r)$, $(x_l,y_l)$ or $(x_r,y_r)$ cannot be improved, for any $k\in[0,1]$. Consider any instance with $x_l=0$, $x_r=1-\epsilon$ for some sufficiently small positive number $\epsilon$, and $y_l=y_r=1$, and the obstacle is between $1-\epsilon$ and $1$. The optimal solution is $(\frac{x_l+x_r}{2},1)$ by Theorem \ref{thm:opt}, and the optimal maximum cost is $\frac{1-\epsilon}{2}+k(1-\frac{1-\epsilon}{2})=\frac{1+k}{2}-\frac{1-k}{2}\epsilon$. However,  the maximum cost induced by  $(x_l,y_r)=(x_l,y_l)=(0,1)$ is $1-\epsilon+k$. We have $\frac{1-\epsilon+k}{(1+k)/2-(1-k)\epsilon/2}\rightarrow 2$, when $\epsilon$ approaches $0$. A symmetric instance shows that the 2-approximation analysis for mechanism $(x_r,y_r)$ is also tight. 


\section{Proof of Theorem \ref{thm:round2}}\label{app:f}

In this proof, we first treat $c$ as a parameter in interval $[0,1]$ and do not specify its value. Then we select the best value of $c$ to minimize the approximation ratio.

\paragraph{Proof of the group strategyproofness.}




Let $(a, b)$ be the output. Note that $a$ and $b$ are independent random variables, and $(1-c)o$ and $o+c(1-o)$ are two constants only related to $k$. Consider a group of agents $S\subseteq N_1\cup N_2$, and let $(a',b')$ be the output when the agents in $S$ misreport. Assume w.l.o.g. that $|a'-a|\ge |b'-b|$ and $|a'-a|>0$.

When $a< x_r$, since $|a'-a|>0$, an agent located at $x_r$ must be in the group and misreport a location to the left of $x_r$, implying that $a'<a<x_r$.  The cost of this agent decreases by at most $(1-k)|b'-b|-(1+k)(a-a')\le 0$, indicating that this agent can never gain. 

When $a=x_r$, it is either the case when an agent at $x_r$  misreports to its left so that $a'<a$, or the case when some agent in $N_1$ misreports to the right of $x_r$ so that $a'>a$. In both cases,  the cost of this agent decreases by at most $(1-k)|b'-b|-(1-k)|a'-a|\le 0$, indicating that this agent can never gain.

\medskip
\paragraph{Proof of the approximation ratio.}
   Given any instance with location profile $\mathbf x$, we 
assume w.l.o.g. that $1-y_r\ge x_l$. By Theorem \ref{thm:opt}, the optimal solution is $(a^*,b^*)=(\frac{x_l+x_r}{2},\frac{y_l-x_l}{2}+\frac12)$, and the optimal maximum cost is 
\begin{align*}
OPT = cost(a^*,b^*,x_l)&=a^*-x_l+k(b^*-a^*)+1-b^*=\frac12[1+(1-k)x_r+k(1-2x_l)-(1-k)y_l].
\end{align*}
We discuss four cases with respect to the output of the mechanism. 

\textbf{Case 1}. $x_r\le o(1-c), y_l \ge o+c(1-o)$. The output is $a=x_r, b=y_l$, and the maximum cost must be achieved by $x_l$, because when $1-y_r\ge x_l$ the cost at $y_r$ is no more than the cost at $x_l$. We have
    \begin{align}
        \frac{cost(a,b,x_l)}{2\cdot OPT} &=\frac{x_r-x_l+k(y_l-x_r)+1-y_l}{1+(1-k)x_r+k(1-2x_l)-(1-k)y_l}\nonumber\\
        &= \frac{1+(1-k)x_r-x_l-(1-k)y_l}{1+(1-k)x_r+k(1-2x_l)-(1-k)y_l}\label{eq:cas10}\\
        &\le  \frac{1+(1-k)(1-c)o-x_l-(1-k)y_l}{1+(1-k)(1-c)o+k(1-2x_l)-(1-k)y_l}\nonumber\\
        &\le \frac{1+(1-k)(1-c)o-x_l-(1-k)(o+c(1-o))}{1+(1-k)(1-c)o+k(1-2x_l)-(1-k)(o+c(1-o))}\nonumber\\
        &=  \frac{1-(1-k)c-x_l}{1+k-(1-k)c-2kx_l}\label{eq:cas11}\\
        &\le  \frac{1-(1-k)c}{1+k-(1-k)c}.\nonumber
    \end{align}
    The first inequality is because \eqref{eq:cas10} is no more than 1, and the last inequality is because \eqref{eq:cas11} is no more than $\frac{1}{2k}$. 

    \textbf{Case 2}. $x_r\ge o(1-c), y_l\le o+c(1-o)$. The output is $a=o(1-c), b=o+c(1-o)$. Note that the cost at $y_r$ is either at most the cost at $y_l$ or at most that at $x_l$ (since $1-y_r\ge x_l$).  Thus the maximum cost is achieved by at least one of the agents at $x_l, x_r, y_l$. First, we consider the cost at $x_l$, and we can assume $x_l\le o(1-c)$; otherwise we have $cost(a,b,x_l)\le cost(a,b,x_r)$ and it reduces to consider the cost at $x_r$.  We have
    \begin{align*}
        \frac{cost(a,b,x_l)}{2\cdot OPT} &=\frac{o(1-c)-x_l+k(o+c(1-o)-o(1-c))+1-o-c(1-o)}{1+(1-k)x_r+k(1-2x_l)-(1-k)y_l}\\
        &= \frac{1-(1-k)c-x_l}{1+(1-k)x_r+k(1-2x_l)-(1-k)y_l}\\
        &\le \frac{1-(1-k)c-x_l}{1+(1-k)(1-c)o+k(1-2x_l)-(1-k)(o+c(1-o))}\\
        &=  \frac{1-(1-k)c-x_l}{1+k(1-2x_l)-(1-k)c}\\
        &\le \frac{1-(1-k)c}{1+k-(1-k)c}.
    \end{align*}

    Second, for the cost at $x_r$, we have
    \begin{align*}
        \frac{cost(a,b,x_r)}{2\cdot OPT} &= \frac{x_r-o(1-c)+kc+1-(o+c(1-o))}{1+(1-k)x_r+k(1-2x_l)-(1-k)y_l}\\
        &=  \frac{x_r-2o(1-c)+1-(1-k)c}{1+(1-k)x_r+k(1-2x_l)-(1-k)y_l}\\
        &\le  \frac{x_r-2o(1-c)+1-(1-k)c}{1+(1-k)x_r+k(1-2x_l)-(1-k)(o+c(1-o))}\\
         &\le  \frac{o-2o(1-c)+1-(1-k)c}{1+(1-k)o+k(1-2x_l)-(1-k)(o+c(1-o))}\\
        &\le  \frac{o-2o(1-c)+1-(1-k)c}{1+(1-k)o+k(1-2\min(o,1-o))-(1-k)(o+c(1-o))}\\
        &\le  \max\left(\frac{1-(1-k)c}{1+k-(1-k)c}, \frac{1+2ck}{2-(1-k)c}, c\right).
    \end{align*}
    The second last inequality is because $x_l\le o$ and $x_l\le 1-y_r\le 1-o$.
    For the last inequality, we regard $o$ as a variable, and it is easy to find that when $0\le o\le 0.5$, we have
    $$
    \frac{cost(a,b,x_r)}{2\cdot OPT}\le \frac{o-2o(1-c)+1-(1-k)c}{1+(1-k)o+k(1-2o)-(1-k)(o+c(1-o))},
    $$
    and when $0.5<o\le 1$, we have
    $$
    \frac{cost(a,b,x_r)}{2\cdot OPT}\le \frac{o-2o(1-c)+1-(1-k)c}{1+(1-k)o+k(1-2(1-o))-(1-k)(o+c(1-o))}.
    $$
    Since both expressions on the right hand side are monotone with respect to $o$ (possibly increasing or decreasing), the upper bound must be attained by the maximum of the three cases when $o=0, 0.5, 1$, establishing the inequality. 

    Last, for the cost at $y_l$ we have
    \begin{align*}
        \frac{cost(a,b,y_l)}{2\cdot OPT} &=\frac{o(1-c)+kc+o+c(1-o)-y_l}{1+(1-k)x_r+k(1-2x_l)-(1-k)y_l}\\
        &\le  \frac{o(1-c)+kc+o+c(1-o)-y_l}{1+(1-k)o(1-c)+k(1-2x_l)-(1-k)y_l}\\
        &\le \frac{o(1-c)+kc+o+c(1-o)-o}{1+(1-k)o(1-c)+k(1-2x_l)-(1-k)o}\\
        &\le \frac{o(1-c)+kc+o+c(1-o)-o}{1+(1-k)o(1-c)+k(1-2\min(o, 1-o))-(1-k)o}\\
        &\le  \max\left(c, \frac{1+2ck}{2-(1-k)c},\frac{1-(1-k)c}{1+k-(1-k)c}\right).
    \end{align*}
For the last inequality, we regard $o$ as a variable, and it is easy to find that when $0\le o\le 0.5$, we have
    $$
    \frac{cost(a,b,y_l)}{2\cdot OPT}\le \frac{o(1-c)+kc+o+c(1-o)-o}{1+(1-k)o(1-c)+k(1-2o)-(1-k)o},
    $$
    and when $0.5<o\le 1$, we have
    $$
    \frac{cost(a,b,y_l)}{2\cdot OPT}\le \frac{o(1-c)+kc+o+c(1-o)-o}{1+(1-k)o(1-c)+k(1-2(1-o))-(1-k)o}.
    $$
    Since both expressions on the right hand side are monotone with respect to $o$, the upper bound must be attained by the maximum of the three cases when $o=0, 0.5, 1$, establishing the inequality. 

 \textbf{Case 3}. $x_r\le o(1-c), y_l\le o+c(1-o)$. The output is $a=x_r, b=o+c(1-o)$, and the maximum cost is achieved by $x_l$ or $y_l$. First, we consider the cost at $x_l$, and we have
    \begin{align*}
        \frac{cost(a,b,x_l)}{2\cdot OPT} &= \frac{x_r-x_l+k(o+c(1-o)-x_r)+1-o-c(1-o)}{1+(1-k)x_r+k(1-2x_l)-(1-k)y_l}\\
        &=  \frac{1+(1-k)x_r-x_l-(1-k)(o+c(1-o))}{1+(1-k)x_r+k(1-2x_l)-(1-k)y_l}\\
        &\le \frac{1+(1-k)(1-c)o-x_l-(1-k)(o+c(1-o))}{1+(1-k)(1-c)o+k(1-2x_l)-(1-k)y_l}\\
        &\le  \frac{1+(1-k)(1-c)o-x_l-(1-k)(o+c(1-o))}{1+(1-k)(1-c)o+k(1-2x_l)-(1-k)(o+c(1-o))}\\
        &= \frac{1-(1-k)c-x_l}{1-(1-k)c+k(1-2x_l)}\\
        &\le   \frac{1-(1-k)c-0}{1-(1-k)c+k(1-0)}\\
        &= \frac{1-(1-k)c}{1+k-(1-k)c}.
    \end{align*}

    Second, for the cost at $y_l$, we have
    \begin{align*}
        \frac{cost(a,b,y_l)}{2\cdot OPT} &=  \frac{x_r+k(o+c(1-o)-x_r)+o+c(1-o)-y_l}{1+(1-k)x_r+k(1-2x_l)-(1-k)y_l}\\
        &=  \frac{(1-k)x_r-y_l+(1+k)(o+c(1-o))}{1+(1-k)x_r+k(1-2x_l)-(1-k)y_l}\\
        &\le  \frac{(1-k)o(1-c)-y_l+(1+k)(o+c(1-o))}{1+(1-k)o(1-c)+k(1-2x_l)-(1-k)y_l}\\
        &\le  \frac{(1-k)o(1-c)-o+(1+k)(o+c(1-o))}{1+(1-k)o(1-c)+k(1-2x_l)-(1-k)o}\\
        &\le \frac{(1-k)o(1-c)-o+(1+k)(o+c(1-o))}{1-(1-k)oc+k(1-2\min(1-o, o(1-c)))}\\
        &\le \max\left(c, \frac{k(2c-c^2)+1-c^2}{2-2c+2ck}, \frac{1-(1-k)c}{1+k-(1-k)c}\right).
    \end{align*}
    The second last inequality is because $x_l\le x_r\le o(1-c)$ and $x_l\le 1-y_r\le 1-o$.
    For the last inequality, we regard $o$ as a variable, and it is easy to find that when $0\le o\le \frac{1}{2-c}$, we have
    $$
    \frac{cost(a,b,y_l)}{2\cdot OPT}\le \frac{o-2oc+c+kc}{1-(1-k)oc+k(1-2o(1-c))},
    $$
    and when $\frac{1}{2-c}<o\le 1$, we have
    $$
    \frac{cost(a,b,y_l)}{2\cdot OPT}\le \frac{o-2oc+c+kc}{1-(1-k)oc+k(1-2(1-o))}.
    $$
    Since both expressions on the right hand side are monotone with respect to $o$, the upper bound must be attained by the maximum of the three cases when $o=0, \frac{1}{2-c}, 1$, establishing the inequality. 

    \textbf{Case 4}. $x_r\ge o(1-c), y_l\ge o+c(1-o)$. The output is $a = o(1-c), b = y_l$, and the maximum cost is achieved by $x_l$ or $x_r$. First, we consider $cost(a,b,x_l)$, and we can assume $x_l\le o(1-c)$, as otherwise $cost(x_l)\le cost(x_r)$. We have
    \begin{align*}
        \frac{cost(a,b,x_l)}{2\cdot OPT} &= \frac{o(1-c)-x_l+k(y_l-o(1-c))+1-y_l}{1+(1-k)x_r+k(1-2x_l)-(1-k)y_l}\\
        &= \frac{1+(1-k)(1-c)o-x_l-(1-k)y_l}{1+(1-k)x_r+k(1-2x_l)-(1-k)y_l}\\
        &\le  \frac{1+(1-k)(1-c)o-x_l-(1-k)y_l}{1+(1-k)o(1-c)+k(1-2x_l)-(1-k)y_l}\\
        &\le  \frac{1+(1-k)(1-c)o-x_l-(1-k)(o+c(1-o))}{1+(1-k)(1-c)o+k(1-2x_l)-(1-k)(o+c(1-o))}\\
        &= \frac{1-x_l-(1-k)c}{1+k(1-2x_l)-(1-k)c}\\
        &\le \frac{1-(1-k)c}{1+k-(1-k)c}.
    \end{align*}

    Second, for the cost at $x_r$, we have
    \begin{align*}
        \frac{cost(a,b,x_r)}{2\cdot OPT} &= \frac{x_r-o(1-c)+k(y_l-o(1-c))+1-y_l}{1+(1-k)x_r+k(1-2x_l)-(1-k)y_l}\\
        &=  \frac{1+x_r-(1+k)(1-c)o-(1-k)y_l}{1+(1-k)x_r+k(1-2x_l)-(1-k)y_l}\\
        &\le  \frac{1+o-(1+k)(1-c)o-(1-k)y_l}{1+(1-k)o+k(1-2x_l)-(1-k)y_l}\\
        &\le  \frac{1+o-(1+k)(1-c)o-(1-k)(o+c(1-o))}{1+(1-k)o+k(1-2x_l)-(1-k)(o+c(1-o))}\\
        &\le  \frac{1-o+2co-(1-k)c}{1+k(1-2\min(o,(1-c)(1-o)))-(1-k)c(1-o)}\\
        &\le  \max\left(\frac{1-(1-k)c}{1+k-(1-k)c}, \frac{k(2c-c^2)+1-c^2}{2-2c+2ck}, c\right).
    \end{align*}
    The second last inequality is because $x_l\le x_r\le o$ and $x_l\le 1-y_r\le 1-y_l\le 1-o-c(1-o)$.
    For the last inequality, we regard $o$ as a variable, and it is easy to find that when $0\le o\le \frac{1-c}{2-c}$, we have
    $$
    \frac{cost(a,b,x_r)}{2\cdot OPT}\le \frac{1-o+2co-(1-k)c}{1+k(1-2o)-(1-k)c(1-o)},
    $$
    and when $\frac{1-c}{2-c}<o\le 1$, we have
    $$
    \frac{cost(a,b,x_r)}{2\cdot OPT}\le \frac{1-o+2co-(1-k)c}{1+k(1-2(1-c)(1-o))-(1-k)c(1-o)}.
    $$
    Since both expressions on the right hand side are monotone with respect to $o$, the upper bound must be attained by the maximum of the three cases when $o=0, \frac{1-c}{2-c}, 1$, establishing the inequality. 

 \medskip   According to the four cases above, the ratio is
    \begin{equation}\label{eq:maxmin}
        \frac{MC(a,b,\mathbf x)}{OPT} \le 2\cdot\max\left(\frac{1-(1-k)c}{1+k-(1-k)c}, \frac{k(2c-c^2)+1-c^2}{2-2c+2ck}, \frac{1+2ck}{2-(1-k)c}, c\right).
    \end{equation}
    We need to select a proper value of $c$ so that the right hand side of \eqref{eq:maxmin} is minimized. Fixing $k$, note that $\frac{1-(1-k)c}{1+k-(1-k)c}$ is decreasing with $c$, and $\frac{1+2ck}{2-(1-k)c}$ is increasing with $c$.
    Consider the equation
    $$
    \frac{1-(1-k)c}{1+k-(1-k)c} = \frac{1+2ck}{2-(1-k)c}.
    $$
    The only solution is
    $$
    c=\frac{1+k^2- \sqrt{k^4-k^3+3k^2+k}}{1-k^2}.
    $$
    Furthermore, when $c=\frac{1+k^2- \sqrt{k^4-k^3+3k^2+k}}{1-k^2}$, we have $ \frac{1-(1-k)c}{1+k-(1-k)c}= \frac{k(2c-c^2)+1-c^2}{2-2c+2ck}$ and $ \frac{1-(1-k)c}{1+k-(1-k)c}\ge c$. 
    Hence, it minimizes the right hand side of \eqref{eq:maxmin}.

\section{Another randomized mechanism for maximum cost}\label{sec:another}

We present another randomized mechanism that is at most 1.441-approximation for maximum cost.

\begin{mechanism}
\label{alg:ran-bad}
    Given location profile $\mathbf x$, let $a$ be $x_r$ and $\frac{x_r}{2}$ with probabilities $\frac{1+k}{3-k}$ and $\frac{2(1-k)}{3-k}$, respectively. Let $b$ be $y_l$ and $\frac{1+y_l}{2}$ with probabilities $\frac{1+k}{3-k}$ and $\frac{2(1-k)}{3-k}$, respectively. Return $(a,b)$.
\end{mechanism}

\begin{lemma}\label{lem:rans}
    Mechanism \ref{alg:ran-bad} is group strategyproof. 
\end{lemma}
\begin{proof}
We consider a group of agents $S\subseteq N_1\cup N_2$. Let $f(\mathbf x)=(a,b)$ be the outcome when all agents report true locations, and $f(\mathbf x_{S}',\mathbf x_{-S})=(a',b')$ be the outcome when the agents in $S$ misreport $\mathbf x_{S}'$, where $a,b,a',b'$ are random variables that follow the distributions given in the mechanism. Assume w.l.o.g. that $|\mathbb E[a]-\mathbb E[a']|\ge |\mathbb E[b]-\mathbb E[b']|$. We show that at least one agent in the group cannot gain by misreporting. 

When $\mathbb E[a']< \mathbb E[a]$, then it must be  $x_r'<x_r$, and the agent located at $x_r$ is in the group. Under the solution $(a,b)$, the cost of the agent at $x_r$ is
$$cost(a,b,x_r)=x_r-\mathbb E[a]+k(\mathbb E[b]-\mathbb E[a])+(1-\mathbb E[b]).$$
Under the solution $(a',b')$, the cost of the agent at $x_r$ is
$$cost(a',b',x_r)=x_r-\mathbb E[a']+k(\mathbb E[b']-\mathbb E[a'])+(1-\mathbb E[b']).$$
Since $|\mathbb E[a]-\mathbb E[a']|\ge |\mathbb E[b]-\mathbb E[b']|$, it follows that 
\begin{align*}
    &cost(a',b',x_r)-cost(a,b,x_r)=(1+k)(\mathbb E[a]-\mathbb E[a'])+(1-k)(\mathbb E[b]-\mathbb E[b'])\ge 0, 
\end{align*}
indicating that this agent cannot gain.

When $\mathbb E[a']> \mathbb E[a]$, there exists  at least one agent $i\in S\cap N_1$. If $\mathbb E[b']\le \mathbb E[b]$, it is clear that any agent located at $[0,\frac{x_r}{2}]$ cannot gain because the change of the endpoints in both regions do not benefit this agent. For an agent $i\in N_1$ located at $(\frac{x_r}{2},x_r]$, under the solution $(a,b)$, the cost of the agent at $x_r$ is
\begin{align*}
    cost(a,b,x_i)&=\frac{1+k}{3-k}(x_r-x_i+k(\mathbb E[b]-x_r))+ \frac{2-2k}{3-k}(x_i-\frac{x_r}{2}+k(\mathbb E[b]-\frac{x_r}{2})) +1-\mathbb E[b]\\
    &=\frac{1+k}{3-k}(x_r-x_i-kx_r)+ \frac{2-2k}{3-k}(x_i-\frac{x_r}{2}-\frac{kx_r}{2}) +1-(1-k)\mathbb E[b]
\end{align*}
Under the solution $(a',b')$, the cost of this agent is
$$cost(a',b',x_i)=\frac{1+k}{3-k}(x_r'-x_i-kx_r')+ \frac{2-2k}{3-k}(|x_i-\frac{x_r'}{2}|-\frac{kx_r'}{2}) +1-(1-k)\mathbb E[b'].$$
Since $\mathbb E[b']\le \mathbb E[b]$, it follows that 
$$cost(a',b',x_i)-cost(a,b,x_i)\ge \frac{1+k}{3-k}(1-k)(x_r'-x_r)-\frac{2-2k}{3-k}\cdot\frac{1+k}{2}(x_r'-x_r)= 0, $$
indicating that this agent cannot gain. 
If $\mathbb E[b']> \mathbb E[b]$, then the agent located at $y_l$ must be in the group and misreport a location on the right of $y_l$. It is easy to see that 
$$cost(a',b',y_l)-cost(a,b,y_l)=(1+k)(\mathbb E[b']-\mathbb E[b])+(1-k)(\mathbb E[a']-\mathbb E[a])\ge 0, $$
and thus this agent cannot decrease the cost.   
\end{proof}

Now we prove the approximation ratio.

\begin{theorem}\label{thm:ranub-bad}
    Mechanism  \ref{alg:ran-bad} is a randomized group strategyproof mechanism. The approximation ratio for maximum cost is  $\frac{4-2k}{3-k}$ when $k\in[0,\kappa]$, and is $\frac{11+2k^3-9k^2}{9+k^2-6k}$ when $k\in[\kappa,1)$, where $\kappa=\frac{9-\sqrt{73}}{4}\approx 0.114$. 
\end{theorem}
\begin{proof} 
    Given any instance with location profile $\mathbf x$, we 
assume w.l.o.g. that $1-y_r\ge x_l$. By Theorem \ref{thm:opt}, the optimal solution is $(a,b)=(\frac{x_l+x_r}{2},\frac{y_l-x_l}{2}+\frac12)$, and the optimal maximum cost is 
\begin{align*}
cost(a,b,x_l)&=a-x_l+k(b-a)+1-b=\frac12[1+(1-k)x_r+k(1-2x_l)-(1-k)y_l].
\end{align*}

    Now we consider the solution  returned by Mechanism  \ref{alg:ran-bad}. We discuss the 4 realizations of the probability distribution. 
    \begin{itemize}
        \item $(x_r,y_l)$ with probability $\frac{(1+k)^2}{(3-k)^2}$.  By the analysis in the proof of Theorem \ref{thm:21k} and the assumption $1-y_r\ge x_l$,  the maximum cost is attained by $x_l$, that is, 
        \begin{align*}
            MC(x_r,y_l,\mathbf x)=cost(x_r,y_l,x_l)&=x_r-x_l+k(y_l-x_r)+1-y_l.
        \end{align*}
        
        \item $(x_r,\frac{y_l+1}{2})$ with probability $\frac{2(1-k)(1+k)}{(3-k)^2}$. The maximum cost is attained by $x_l$ or $y_l$. The cost of $x_l$ is $x_r-x_l+k(\frac{y_l+1}{2}-x_r)+1-\frac{y_l+1}{2}$, and the cost of $y_l$ is $\frac{1-y_l}{2}+k(\frac{y_l+1}{2}-x_r)+x_r$.  It is easy to see that the cost of $y_l$ is no less than the cost of $x_l$. Hence, the maximum cost is attained by $y_l$, that is, 
         \begin{align*}
            MC(x_r,\frac{y_l+1}{2},\mathbf x)=cost(x_r,\frac{y_l+1}{2},y_l)&=\frac{1-y_l}{2}+k(\frac{y_l+1}{2}-x_r)+x_r.
        \end{align*}
        
        \item $(\frac{x_r}{2},y_l)$ with probability $\frac{2(1-k)(1+k)}{(3-k)^2}$. The maximum cost is attained by $x_r$, that is, 
         \begin{align*}
            MC(\frac{x_r}{2},y_l,\mathbf x)=cost(\frac{x_r}{2},y_l,x_r)&=\frac{x_r}{2}+k(y_l-\frac{x_r}{2})+1-y_l.
        \end{align*}
        
         \item $(\frac{x_r}{2},\frac{1+y_l}{2})$ with probability $\frac{4(1-k)^2}{(3-k)^2}$. The maximum cost is attained by $x_r$ or $y_l$, where both costs are equal to  
         $$MC(\frac{x_r}{2},\frac{1+y_l}{2},\mathbf x)=\frac{x_r}{2}+k(\frac{1+y_l}{2}-\frac{x_r}{2})+\frac{1-y_l}{2}.$$

    \end{itemize}
Therefore, the expected maximum cost of the solution returned by the mechanism is
\begin{align*}
    &\frac{(1+k)^2}{(3-k)^2}\cdot \left(x_r-x_l+k(y_l-x_r)+1-y_l\right)+ \frac{2(1-k)(1+k)}{(3-k)^2}\cdot \left(\frac{1-y_l}{2}+k(\frac{y_l+1}{2}-x_r)+x_r\right)\\
    &+ \frac{2(1-k)(1+k)}{(3-k)^2}\cdot \left(\frac{x_r}{2}+k(y_l-\frac{x_r}{2})+1-y_l\right)+ \frac{4(1-k)^2}{(3-k)^2}\cdot \left( \frac{x_r}{2}+k(\frac{1+y_l}{2}-\frac{x_r}{2})+\frac{1-y_l}{2}\right)\\
    =~& \frac{(1+k)^2}{(3-k)^2}\cdot \left(x_r-x_l+k(y_l-x_r)+1-y_l\right)+ \frac{2(1-k)(1+k)}{(3-k)^2}\cdot \left(\frac{3-3y_l}{2}+k(\frac{3y_l+1}{2}-\frac{3x_r}{2})+\frac{3x_r}{2}\right)\\
    &+ \frac{4(1-k)^2}{(3-k)^2}\cdot \left( \frac{x_r}{2}+k(\frac{1+y_l}{2}-\frac{x_r}{2})+\frac{1-y_l}{2}\right)\\
     =~& \frac{(1+k)^2}{(3-k)^2}\cdot \left(x_r-x_l+k(y_l-x_r)+1-y_l\right)+ \frac{2(1-k)(1+k)}{(3-k)^2}\cdot \left(\frac{3(k-1)(y_l-x_r)}{2}+\frac{3+k}{2} \right)\\
    &+ \frac{4(1-k)^2}{(3-k)^2}\cdot \left( \frac{x_r}{2}+k(\frac{1+y_l}{2}-\frac{x_r}{2})+\frac{1-y_l}{2}\right)\\
    =~& \frac{(1+k)(2-k)}{3-k}+\frac{2(1-k)(3-k)x_r -2(1-k)(3-k)y_l -(1+k)^2x_l}{(3-k)^2}\\
     =~& \frac{(1+k)(2-k)+2(1-k)(x_r-y_l)}{3-k}-\frac{(1+k)^2x_l}{(3-k)^2}.
\end{align*}

Then, the ratio between the expected maximum cost and the optimal maximum cost is
\begin{equation}\label{eq:exc}
    2\cdot \frac{ \frac{(1+k)(2-k)+2(1-k)(x_r-y_l)}{3-k}-\frac{(1+k)^2x_l}{(3-k)^2}}{1+(1-k)x_r+k(1-2x_l)-(1-k)y_l}.
\end{equation}
Considering $x_r-y_l$ as a variable of the function in \eqref{eq:exc},  the derivative with respect to this variable is always non-negative for any $k\in[0,1)$, which implies that the maximum possible value is achieved when $x_r=y_l$. 
Then the ratio becomes 
\begin{equation}\label{eq:tgd}
    2\cdot \frac{ \frac{(1+k)(2-k)}{3-k}-\frac{(1+k)^2x_l}{(3-k)^2}}{1+k-2kx_l}=\frac{2}{(3-k)^2}\cdot \frac{(1+k)(2-k)(3-k)-(1+k)^2x_l}{1+k-2kx_l}.
\end{equation}
Let $\kappa=\frac{9-\sqrt{73}}{4}\approx 0.114$ be the root of the equation $\frac{(1+k)(2-k)(3-k)}{1+k}=\frac{(1+k)^2}{2k}$. 
\begin{itemize}
    \item When $k\in[0,\kappa]$, we have $\frac{(1+k)(2-k)(3-k)}{1+k}\le \frac{(1+k)^2}{2k}$, the maximum value of the ratio in \eqref{eq:tgd} is achieved when $x_l=0$, that is 
    \begin{equation}\label{eq:tgd33}
    \frac{2}{(3-k)^2}\cdot \frac{(1+k)(2-k)(3-k)}{1+k}=\frac{4-2k}{3-k}.
    \end{equation}

    \item When $k\in[\kappa,1]$, we have $\frac{(1+k)(2-k)(3-k)}{1+k}\ge \frac{(1+k)^2}{2k}$, and the ratio in  \eqref{eq:tgd} is increasing with $x_l$. Since $1-y_r\ge x_l$, $x_l$ is upper bounded by $\frac12$. Letting $x_l=\frac12$, the maximum value of the ratio in \eqref{eq:tgd} is 
    \begin{equation}\label{eq:tgd44}
    \frac{2(1+k)(2-k)(3-k)-(1+k)^2}{(3-k)^2}=\frac{11+2k^3-9k^2}{9+k^2-6k}.
    \end{equation}
\end{itemize}
\end{proof}

The maximum possible value of the approximation ratio over all $k\in[0,1)$ is $9-6\sqrt[3]{2}\approx 1.441$, which is attained by $k=3-2\sqrt[3]{2}$. Hence, generally we can say that 
 Mechanism  \ref{alg:ran-bad} is $1.441$-approximation for any $k\in[0,1)$, and in particular, it is $\frac43$-approximation when $k=0$, and nearly optimal when $k$ approaches 1. 
Compared with the approximation ratio $\frac{2}{1+k}$ of the deterministic \textsc{TwoInnerExtreme}, this randomized one improves when $k\le 0.396$, but is worse for any larger $k$.


\section{Lower Bound by Experiment}\label{app:experiment}

In this appendix, we provide the details of our computational experiment to establish the lower bound function $L(k)$. We include the complete algorithm, convergence analysis, and the resulting data.

\subsection{Experimental Setup}

We consider the two-facility location problem with $n=4$ agents. For each $k \in [0,1)$, we search for the worst-case approximation ratio by enumerating all possible combinations of agent positions that maintain the same facility placement $(a_0,b_0)$ under any deterministic strategyproof mechanism. As proved in Theorem~\ref{thm:lower}, there are exactly $2^4 = 16$ such profiles for each candidate $(a_0,b_0)$.

For each $k \in \{0, 0.01, 0.02, \ldots, 0.99\}$ and $o \in [0.5,1]$, we compute:
\[
r^*(k,o) = \min_{\substack{a_0 \in [0,o] \\ b_0 \in [o,1]}} \max_{\text{16 profiles}} \frac{\text{cost}((a_0,b_0),\text{profile})}{\text{OPT}(\text{profile})}
\]
where $\text{OPT}(\text{profile})$ denotes the optimal cost for the given profile.

The overall lower bound is then:
\[
L(k) = \max_{o \in [0.5,1]} r^*(k,o).
\]

\subsection{Numerical Method and Convergence Analysis}

We discretize the parameter space $[0,o] \times [o,1]$ using a uniform grid of $1000 \times 1000$ points. The choice of grid density is justified by a convergence analysis comparing results from $300\times300$ and $1000\times1000$ grids. The maximum difference between the two grids is $\Delta = 0.0021$ at $k=0.01$.

Under the standard grid search error model $f_N = f^* + C N^{-p}$, where $f_N$ is the value on an $N\times N$ grid, $f^*$ is the true optimum, and $p$ is the convergence order, we estimate $p \approx 1.5$ based on the smoothness of the resulting function $L(k)$. This yields an error bound of $0.00062$ for the $1000\times1000$ grid, ensuring three decimal places of accuracy.

\subsection{Complete Results}

Table~\ref{tab:lower-bound-data} provides the complete lower bound function $L(k)$ for $k \in [0,1)$. These values were computed using the algorithm described below.

\begin{table}[ht]
\centering
\caption{Lower bound $B(k)$ for strategyproof mechanisms}\label{tab:lower-bound-data}
\begin{tabular}{ccccccccccc}
\toprule
$k$ & $B(k)$ & $k$ & $B(k)$ & $k$ & $B(k)$ & $k$ & $B(k)$ & $k$ & $B(k)$ \\
\midrule
0.00 & 2.000000 & 0.20 & 1.428571 & 0.40 & 1.295660 & 0.60 & 1.195596 & 0.80 & 1.099201 \\
0.01 & 1.819950 & 0.21 & 1.420010 & 0.41 & 1.290144 & 0.61 & 1.190880 & 0.81 & 1.094324 \\
0.02 & 1.756491 & 0.22 & 1.411843 & 0.42 & 1.284899 & 0.62 & 1.186074 & 0.82 & 1.089445 \\
0.03 & 1.711085 & 0.23 & 1.403846 & 0.43 & 1.279744 & 0.63 & 1.181276 & 0.83 & 1.084476 \\
0.04 & 1.676035 & 0.24 & 1.395953 & 0.44 & 1.274488 & 0.64 & 1.176489 & 0.84 & 1.079593 \\
0.05 & 1.646055 & 0.25 & 1.388656 & 0.45 & 1.269273 & 0.65 & 1.171711 & 0.85 & 1.074628 \\
0.06 & 1.620805 & 0.26 & 1.381262 & 0.46 & 1.264283 & 0.66 & 1.166856 & 0.86 & 1.069737 \\
0.07 & 1.598661 & 0.27 & 1.374163 & 0.47 & 1.259114 & 0.67 & 1.162002 & 0.87 & 1.064773 \\
0.08 & 1.578814 & 0.28 & 1.367500 & 0.48 & 1.254126 & 0.68 & 1.157248 & 0.88 & 1.059809 \\
0.09 & 1.560757 & 0.29 & 1.360748 & 0.49 & 1.249178 & 0.69 & 1.152496 & 0.89 & 1.054849 \\
0.10 & 1.544260 & 0.30 & 1.354144 & 0.50 & 1.244142 & 0.70 & 1.147611 & 0.90 & 1.049915 \\
0.11 & 1.529331 & 0.31 & 1.347719 & 0.51 & 1.239288 & 0.71 & 1.142843 & 0.91 & 1.044915 \\
0.12 & 1.515347 & 0.32 & 1.341455 & 0.52 & 1.234303 & 0.72 & 1.138046 & 0.92 & 1.039947 \\
0.13 & 1.501947 & 0.33 & 1.335341 & 0.53 & 1.229417 & 0.73 & 1.133207 & 0.93 & 1.034974 \\
0.14 & 1.489905 & 0.34 & 1.329362 & 0.54 & 1.224584 & 0.74 & 1.128325 & 0.94 & 1.029998 \\
0.15 & 1.478327 & 0.35 & 1.323508 & 0.55 & 1.219637 & 0.75 & 1.123576 & 0.95 & 1.024991 \\
0.16 & 1.467175 & 0.36 & 1.317768 & 0.56 & 1.214906 & 0.76 & 1.118698 & 0.96 & 1.020008 \\
0.17 & 1.457173 & 0.37 & 1.312133 & 0.57 & 1.210061 & 0.77 & 1.113822 & 0.97 & 1.015003 \\
0.18 & 1.447090 & 0.38 & 1.306595 & 0.58 & 1.205185 & 0.78 & 1.108948 & 0.98 & 1.010001 \\
0.19 & 1.437579 & 0.39 & 1.301144 & 0.59 & 1.200331 & 0.79 & 1.104075 & 0.99 & 1.005000 \\
\bottomrule
\end{tabular}
\end{table}

\subsection{Implementation Details}

The core algorithm is implemented in Python and consists of the following components:

\begin{enumerate}
\item \textbf{Profile generation}: For each candidate $(a_0,b_0)$, generate all $16$ profiles according to the scheme in Theorem~\ref{thm:lower}.
\item \textbf{Cost calculation}: Compute the maximum cost for the mechanism placing facilities at $(a_0,b_0)$ and the optimal cost for each profile.
\item \textbf{Ratio computation}: Calculate the approximation ratio as the maximum of the $16$ cost ratios.
\item \textbf{Grid search}: Minimize this ratio over all $(a_0,b_0)$ in a $1000\times1000$ grid covering $[0,o] \times [o,1]$.
\item \textbf{Worst-case $o$}: For each $k$, maximize the minimum ratio over $o \in [0.5,1]$.
\end{enumerate}

The complete code is shown below.

\begin{python}
import math
import itertools

# =====================
# Parameters
# =====================
K_GRID = [i/100 for i in range(0, 100)]      # k in [0,1)
O_GRID = [0.5 + i/40 for i in range(0, 21)]  # o in [0.5,1]
# In fact, the worst case always happens when o is 0.5, therefore we can use O_GRID = [0.5] instead.
A_GRID_SIZE = 1000
B_GRID_SIZE = A_GRID_SIZE
EPS = 1e-8

# =====================
# Cost functions
# =====================
def cost_left(x, a, b, k):
    return abs(x - a) + k * (b - a) + (1 - b)

def cost_right(x, a, b, k):
    return abs(x - b) + k * (b - a) + a

def max_cost(a, b, profile, k):
    xl, xr, yl, yr = profile
    return max(
        cost_left(xl, a, b, k),
        cost_left(xr, a, b, k),
        cost_right(yl, a, b, k),
        cost_right(yr, a, b, k)
    )

# =====================
# Optimal MC (oracle)
# =====================
def opt_max_cost(profile, k):
    xl, xr, yl, yr = profile
    if 1 - yr >= xl:
        a = (xl + xr) / 2
        b = (yl - xl) / 2 + 0.5
    else:
        a = (xr - yr) / 2 + 0.5
        b = (yl + yr) / 2
    return max_cost(a, b, profile, k)

# =====================
# 16 profiles generator
# =====================
def generate_profiles(a0, b0, o, eps=1e-6):
    choices = [
        [0, a0],            # xl
        [a0, o - eps],      # xr
        [o + eps, b0],      # yl
        [b0, 1]             # yr
    ]
    return list(itertools.product(*choices))

# =====================
# Worst ratio for fixed (a0,b0)
# =====================
def worst_ratio(a0, b0, o, k):
    profiles = generate_profiles(a0, b0, o)
    worst = 0
    for p in profiles:
        ALG = max_cost(a0, b0, p, k)
        OPT = opt_max_cost(p, k)
        if OPT < EPS:
            continue
        worst = max(worst, ALG / OPT)
    return worst

# =====================
# Main experiment
# =====================
def run_experiment():
    results = {}

    for k in K_GRID:
        lb_k = 2 / (1 + math.sqrt(k)) if k > 0 else 2
        worst_over_o = lb_k
        o_result = {}

        for o in O_GRID:
            best = float("inf")
            for i in range(A_GRID_SIZE):
                a0 = o * i / A_GRID_SIZE
                for j in range(B_GRID_SIZE):
                    b0 = o + (1 - o) * j / B_GRID_SIZE
                    r = worst_ratio(a0, b0, o, k)
                    best = min(best, r)

            best = max(best, lb_k)   # respect analytical LB
            o_result[o] = best
            worst_over_o = max(worst_over_o, best)

        results[k] = worst_over_o

        # one-line log per k
        print(
            f"k={k:.3f} | "
            + " ".join([f"o={o:.2f}:{o_result[o]:.3f}" for o in O_GRID])
            + f" || worst={worst_over_o:.3f}"
        )

    return results

# Final output
if __name__ == "__main__":
    res = run_experiment()
    print(res)

\end{python}

\end{document}